\documentclass[prodmode]{acmsmall-ec14}

\usepackage[ruled]{algorithm2e}
\usepackage{amsmath, amssymb}
\usepackage{comment}
\usepackage{epstopdf}
\renewcommand{\algorithmcfname}{ALGORITHM}
\SetAlFnt{\small}
\SetAlCapFnt{\small}
\SetAlCapNameFnt{\small}
\SetAlCapHSkip{0pt}
\IncMargin{-\parindent}

\acmVolume{9}
\acmNumber{4}
\acmArticle{39}
\acmYear{2013}
\acmMonth{6}
\usepackage[numbers]{natbib}
\newcommand{\bR}{\mathbb{/R}}

\newcommand{\dom}{\mathrm{dom}}
\newcommand{\relint}{\mathrm{relint }\,}
\newcommand{\bdd}{\mathrm{bdd \,}}

\begin{document}

\markboth{Bachrach et al.}{Optimising Trade-offs Among Stakeholders in Ad Auctions}

\title{Optimising Trade-offs Among Stakeholders in Ad Auctions}
\author{
YORAM BACHRACH
\affil{Microsoft Research}
SOFIA CEPPI
\affil{Microsoft Research}
IAN A. KASH
\affil{Microsoft Research}
PETER KEY
\affil{Microsoft Research}
DAVID KUROKAWA
\affil{Carnegie Mellon University}}

\begin{abstract}
We examine trade-offs among  stakeholders in ad auctions. Our metrics are the revenue for the utility of the auctioneer, the number of clicks for the utility of the users and the welfare for the utility of the advertisers. We show how to optimize linear combinations of the stakeholder utilities, showing that these can be tackled through a GSP auction with a per-click reserve price. We then examine constrained optimization of stakeholder utilities. 

We use simulations and analysis of real-world sponsored search auction data to demonstrate the feasible trade-offs, examining the effect of changing the allowed number of ads on the utilities of the stakeholders. We investigate both short term effects, when the players do not have the time to modify their behavior, and long term equilibrium conditions.

Finally, we examine a combinatorially richer constrained optimization problem, where there are several possible allowed configurations (templates) of ad formats. This model captures richer ad formats, which allow using the available screen real estate in various ways. 
We show that two natural generalizations of the GSP auction rules to this domain are poorly behaved, resulting in not having a symmetric Nash equilibrium or having one with poor welfare. We also provide positive results for restricted cases.
\end{abstract}



\keywords{Sponsored search; keyword auctions; generalized second price; reserve price; trade-offs}


\begin{bottomstuff}
Author's addresses: Y. Bachrach, S. Ceppi, I. Kash, and P. Key Microsoft Research Cambridge; D. Kurokawa, Computer Science Department, Carnegie Mellon University.
\end{bottomstuff}

\maketitle


\section{Introduction}

Studies of auctions tend to focus on one particular objective, with the goal of optimizing it.  For example, the VCG auction maximizes welfare, while Myerson~\citeyear{myerson1981optimal} studied the design of auctions that maximize revenue in a single parameter setting.  However, for designers of practical auctions, such pure objectives are rarely the goal.  Instead, the auction designer, even if caring only about long-term revenue, must take into account various aspects of the health of the marketplace.

A notable example is that of search advertising auctions, where the opportunity to advertise alongside search results is auctioned off.  There, the search engine wants to maximize revenue, but faces competition from both other search engines and other advertising outlets. As a result, maximizing revenue requires striking a balance between extracting revenue from current advertisers and keeping them or attracting more.  A classic result of Bulow and Klemperer~\citeyear{bulow96} shows that, in some cases, attracting even a single additional bidder can be as valuable as perfectly optimizing revenue.  Given this, the auctioneer should take at least some consideration of the welfare of advertisers.

At the same time, a search engine is a two-sided platform where the good being sold to advertisers is the attention of searchers.  Decisions that maximize revenue could result in a poor user experience if more or worse ads are shown.  As searchers have other options, the supply of goods to be sold depends on how satisfied they are.  Thus, the auctioneer should take their welfare into account as well.

While auction theory provides a rich set of tools for optimizing the welfare of a single group, much less attention has been paid to optimizing tradeoffs among multiple stakeholders.  Better understanding this, in the particular context of sponsored search, is the goal of this paper.

\subsection{Our Contribution}

We begin by analyzing the optimization of a linear combination of objectives relevant to the search engine (revenue), advertisers (welfare), and users (clicks).  We provide an auction design that optimizes any objective that is {\em linear} in the click probability of an advertiser, of which each of these three objectives (and any linear combination of them) is a special case.  Like current practice, this auction uses a rank score to order ads and prices can be computed in a truthful or Generalized Second Price (GSP) fashion.  Importantly, this auction achieves trade-offs between revenue and other objectives through the use of a per-click reserve.

Next, we approach the same problem through the lens of constrained optimization.  Rather than optimize a particular combination, we could constrain some of the objectives to lie above some minimal value while optimizing the remaining ones.  This may be a more natural approach in practice, since it is not clear how the various objectives should be weighted, but given a current state of the marketplace it is natural to ask how to  make the system one percent better for users while holding revenue neutral.  Perhaps unsurprisingly, we are able to show that this problem is equivalent to the unconstrained problem via Lagrangean duality (the analysis is not trivial because the object we seek is an optimal {\em function}).

Of course, these three objectives are not the only natural ones for an auction designer to consider.  For example, using the number of clicks (click yield) treats a deceptive advertisement that tricks people into clicking on it as beneficially relevant for users.  An alternative that fits within our constrained optimization framework is to constrain the number of ads shown (impression yield), so that only the ``best'' are shown to users.  We characterize the design of the optimal option with this constraint.  Interestingly, the optimal way to ensure this constraint is satisfied through a per-impression reserve, which previous work has identified as a poor tool from a pure revenue maximization standpoint~\cite{roberts13,thompson13}.

To complement this theoretical exploration, we follow~\cite{roberts13} and examine the relative performance of the optimal mechanism to various designs that have previously been considered through extensive experiments on both synthetic and real data. The experiments on synthetic data confirm the theoretical results: a ranking algorithm with a per-impression reserve usually outperforms the other evaluated ranking algorithms in terms of revenue, welfare, and click yield. For the second part of the experiments, we use historical data from Microsoft Bing for a keyword with over 500 bidders and for a keyword with fewer than 10 bidders. In the former case, all the ranking algorithms perform similarly for high values of impression yield while for low values of impression yield the algorithm with a per-impression reserve provides a higher revenue. For the latter case, the algorithm with the per-impression reserve provides an essentially constant revenue while controlling the number of ads shown.  This is higher for low numbers of ads and lower for high numbers of ads. Finally, on a sample from a week’s worth of data across all keywords on Bing, we observe opposite results when trading off between click and impression versus revenue and impression.

Finally, we consider the trade-offs created due to the existence of several different types of ads. For example, instead of displaying several text ads the search engine may at times wish to display one large image-based ad or a block of ads for different products that can be purchased. This setting diverges from classical theory in two ways. First, the set of advertisers in the two camps are not necessarily equivalent (e.g. an advertiser may only desire a placement for image-based ads). Second, the search engine's available slots are no longer pre-set and can change based off the bids. We will then show that the truthful auction in this setting is largely analogous to the standard setting, but attempting to utilize GSP payment rules results in significant complications.  Our results here are quite negative, with GSP being poorly behaved except in restricted cases.  As both Google and Bing use GSP, identifying a solution to this problem is an important direction for future work.

\subsection{Related Work}

Our work sits in a long line of papers that study generalized second price auctions for sponsored search~\cite{varian07,eos07}.  Ostrovsky and Schwarz~\citeyear{ostrovsky11} studied the effects of applying Myerson's~\citeyear{myerson1981optimal} optimal per-click reserve on Yahoo!  (For practical reasons they actually implemented per-impression reserve.)  Lahaie and Pennock~\citeyear{lahaie07} proposed the idea of a squashing parameter to improve revenue, and Lahaie and McAfee~\citeyear{lahaie11} showed that it can increase welfare as well.  Thompson and Leyton-Brown~\citeyear{thompson13} studied a variety of ways of increasing revenue, including through a reserve of the optimal form.

In the study of trade-offs, Likhodedov and Sandholm~\citeyear{likhodedov03} showed how to optimize a combination of revenue and welfare. Diakonikolas et al.~\citeyear{Diakonikolas2012} explore the trade-off revenue/welfare focusing the attention on computational complexity issues for both exact and approximated deterministic mechanisms. The convex combination of revenue and welfare has also been considered in a study about refinements in the prediction of the relevance of ads done by Sundararajan and Talgam-Cohen~\citeyear{Sundararajan2013}.  Liu and Chen~\citeyear{liu06} and Liu et al.~\citeyear{liu10} compared the designs of revenue-optimal and welfare-optimal auctions in a simple setting.  Edelman and Schwarz~\citeyear{edelman10} argued that setting an optimal reserve price will lead to a significant revenue gain with minimal welfare loss.  Athey and Ellison~\citeyear{athey11} studied a model of consumer search and showed that reserve prices can increase user welfare.  Roberts et al.~\citeyear{roberts13} studied the revenue optimal auction and showed that empirically it led to good trade-offs between revenue and other objectives, but did not have a theoretical explanation for this.  They also showed that symmetric Nash equilibra continue to exist when per-impression reserves, per-click reserves of the optimal form, and other linear alterations to the ranking are performed.  

Our motivations for considering trade-offs are driven by issues of endogenous participation.  Jehiel and Lamy~\citeyear{jehiel13} examine a general model of this problem in auctions, and their application of their results to sponsored search settings yields a qualitatively similar result to ours: reserves prices should be set, but not at as high a level as in a setting where participation is exogenous.

\section{Preliminaries}

We now describe the setting  and notation we use throughout the paper.  
Our setting is a standard sponsored search ad auction or a position auction, where there are $n$ bidders or advertisers;  bidder $i$ submits a bid $b_i$ for their ad to be displayed, and a displayed ad that is clicked is charged a price $p_i$. 

We use the following notation:
\begin{itemize}
\item $b = (b_1, b_2, ..., b_n)$: the vector of bids.
\item $t_i \in [c_i, d_i]$: the type of bidder $i$ (their true value of a click),  which can take a value in the compact set $[c_i, d_i]$.
\item $T$: the set of valid bidder types/bids vectors (i.e. $b, (t_1, t_2, ..., t_n) \in A$).
\item Types are independently distributed with density function $f_i$ for bidder $i$,  giving density $f(t)=\prod_i f_i (t_i)$.
\item  $\phi_i : \mathbb{R}_{\geq 0} \rightarrow \mathbb{R}$, the associated virtual value function of the bidder $i$, assumed to be regular.  I.e. $\phi_i(t_i) = t_i - (1 - F(t_i))/f(t_i)$ is non-decreasing.
\item $w_i$: the ad effect of bidder $i$.
\item $x_i : T \rightarrow \mathbb{R}^n$: the allocation rule chosen. This represents the expected slot effect given to bidder $i$.
\item $t_iw_ix_i$ the expected welfare of advertiser $i$, determined as his value for a click multiplied by his probability of a click (factored into an ad effect and a slot effect).
\item $p_i : T \rightarrow \mathbb{R}^n$: the payment function of bidder $i$.
\end{itemize}
As we will be looking at linear objectives and constraints, we introduce:
\begin{itemize}
\item  $OBJ = \alpha\ \text{revenue} + \beta\ \text{welfare} + \gamma\ \text{click yield}$: the objective we will frequently maximize in our auctions. 
We assume $\alpha, \beta, \gamma \geq 0$ and $\alpha + \beta > 0$ (the case of $\alpha = \beta = 0$ is largely trivial and due to its cumbersome edge-case nature will not be considered).
\item $\psi_i(z) = \alpha \phi_i(z) + \beta z + \gamma$: here the $\alpha, \beta, \gamma$ are equivalent to those used in $OBJ$.
\end{itemize}

To define an auction mechanism the designer must specify both an allocation rule (or ranking algorithm) and a payment function.
We use ranking algorithms based on the parameters described above that are {\em monotone}, i.e., the expected valuation of each advertiser does not decrease when their type increases, and payments that provide the ``right'' incentives to advertisers to report bids that correspond to their types. We can obtain such payments in two ways: a) since the ranking algorithm is monotone, we can compute the payments as prescribed by Myerson~\citeyear{myerson1981optimal}, yielding a truthful mechanism; b) We can use GSP payments. In the latter case, we need a solution concept that characterizes a {\em single} expected outcome, and apply the commonly used one:  \emph{symmetric} or \emph{locally envy free} Nash equilibrium (SNE)~\cite{eos07,varian07}. An SNE is a refinement of a Nash equilibrium requiring that for every pair of advertisers $i$ and $j$, the following constraint is satisfied:
\[
x_j(b)\left(t_i - \frac{w_j}{w_i}p_j(b)\right) \leq x_i(b)(t_i - p_i(b))
\]
Intuitively, this says that $i$ prefers their slot and payment to $j$'s slot and payment.
Previous work has shown that with GSP payments the ``lowest'' SNE corresponds to the truthful outcome, in the sense that each bidder ends up with the same slot and payment, for a broad class of allocation rules~\cite{aggarwal06,eos07,varian07,roberts13}.

\section{Unconstrained Optimization}
\label{sec:linear_obj}

We first begin with auctions that solve an (essentially) unconstrained optimization problem%
\footnote{Strictly speaking there is a constraint that the allocation rule be monotone, but for the cases we examine the optimal solution has this property so we can ignore this constraint and perform unconstrained optimization.}.
The natural problem of this form is to assign advertisers to slots to maximize some combination of objectives such as revenue, welfare, and click yield.  A key observation, previously used in a different setting by Likhodedov and Sandholm~\citeyear{likhodedov03}, is that all of these objectives are linear in the allocation.  Thus they can naturally be folded together to yield an optimal auction by using a ``rank score'' that takes the appropriate combination of them.  The following theorem captures this for our objective function of the form $OBJ = \alpha\ \text{revenue} + \beta\ \text{welfare} + \gamma\ \text{click yield}$ (where $\alpha,\beta,\gamma \geq 0$). This allows trading off measures of the utility of the search platform, the advertisers, and the users.  Of course, any other linear term could easily be included.

\begin{theorem}
\label{thm:opt}
The auction that maximizes $OBJ$ maximizes the equivalent objective
\begin{equation*}
\int_T \sum_{i = 1}^n w_i \psi_i(t_i) x_i(t) f(t) dt.
\end{equation*}
Thus, the truthful auction that maximizes for any instance of types/bids $t$, the expression $\sum_{i = 1}^n w_i \psi_i(t_i) x_i(t)$ is optimal.
\end{theorem}
\begin{proof}
We first note that $\alpha, \beta, \gamma$ are all non-negative and that each of the function components are monotone, resulting in monotonicity. 

Let us now consider each of the three terms separately.
\begin{enumerate}
\item revenue:
following Myerson~\citeyear{myerson1981optimal}, we have that 
\begin{align*}
\int_T \sum_{i = 1}^n w_i p_i(t) f(t) dt
= - \sum_{i = 1}^n w_i U_i(x, p, c_i) + \int_T \sum_{i = 1}^n w_i \phi_i(t_i) x_i(t) f(t) dt
\end{align*}
where $U_i$ is the expected utility of bidder $i$ when bidding $c_i$ in the auction described by $x$ and $p$ as in~\cite{myerson1981optimal}.

\item welfare:
\begin{align*}
\int_T \sum_{i = 1}^n w_i t_i x_i(t) f(t) dt
\end{align*}

\item click yield:
\begin{align*}
\int_T \sum_{i = 1}^n w_i x_i(t) f(t) dt
\end{align*}
\end{enumerate}
Combining the above, we get that:
\begin{align*}
OBJ
&= - \alpha\sum_{i = 1}^n w_i U_i(x, p, c_i) + \int_T \sum_{i = 1}^n w_i \left(\alpha \phi_i(t_i) + \beta t_i + \gamma\right) x_i(t) f(t) dt.
\end{align*}
Again following~\cite{myerson1981optimal}, the first term can be set to $0$ by setting the payment $p$ accordingly, and as feasible direct revelation mechanisms must have that $U_i(p, x, c_i) \geq 0$, we have that this is optimal. Therefore, to optimize $OBJ$, we need only optimize the second term.
\end{proof}

Note that similarly to the revenue-only case, Theorem~\ref{thm:opt} can clearly correspond to a ranking based on $w_i \psi_i(t_i)$. In the case that the distributions across types are identical, the $\phi_i$ are equivalent to some $\phi$, and in turn the $\psi_i$ are equivalent to some $\psi$. A single bidder-independent reserve price $r$ can then be implemented where $r = \psi^{-1}(0)$.

\subsection{Implications for GSP}
\label{sec:linear-approx}

The analysis in Theorem~\ref{thm:opt} assumes that advertising is sold using a truthful auction.  However, the major sponsored search platforms, such as Google and Bing's ad platforms, implement auctions that are {\em not} truthful.  Instead of a truthful payment rule, these platforms use the GSP payment rule (that is, a bidder pays the value of the minimum bid required to retain the slot allocated to them). For many purposes this distinction is not important, as there exists an equilibrium of the GSP auction that implements the optimal truthful outcome, in the sense that the allocations and payments for all advertisers are the same~\cite{eos07,varian07,roberts13}.  However, this analysis has been shown to work for rankings which are {\em linear} in bids (see in particular \cite{roberts13}).  This is important because, even in the regular case, the virtual valuation functions $\phi_i$ need not be linear, and it is not clear that this good behavior extends to this setting.  Nevertheless, Theorem~\ref{thm:opt} can at the very least be applied to optimizing linear approximation of the $\phi_i$.  Thompson and Leyton-Brown~\citeyear{thompson13} observed that $\phi_i$ is in fact linear for uniform value distributions. In fact, many Beta distributions are also close to linear for much of their range.


\section{Constrained Optimization}
\label{sec-constrained}

Section~\ref{sec:linear_obj} allowed linear trade-offs between the utility measures of the different interested parties, but did not allow an objective function that requires a minimal threshold utility for them. We now consider objective functions that allow setting such a minimal threshold. One example is designing a revenue optimal auction under the constraint that the social welfare exceeds a minimal threshold value of $\theta$.   More generally, we allow several constraints  of the form $\alpha_k\ \text{revenue} + \beta_k\ \text{welfare} + \gamma_k\ \text{click yield} \geq  \theta_k$.
 
Note  that the equivalent objective function derived in Theorem \ref{thm:opt}  is a linear functional on $x = (x_1, x_2, ..., x_n)$, which we can write in succinct notation as
\begin{align*}
G(x) = \int_T \sum_{i = 1}^n w_i \psi_i(t_i) x_i(t) f(t) dt.
\end{align*}
By analogy,  it follows from Theorem \ref{thm:opt}  and its proof that each   constraint can be expressed in the form $a_k (x)\geq \theta_k$ where $a_k$ is a linear functional and hence we can represent a set of constraints in the succinct form  $A(x) + Z = \theta$, where $\theta= (\theta_1, \theta_2, ..., \theta_r)$,  $Z= \{z \in \bR^{\rho(A)}: z\leq 0\} $  is the non-positive cone in the case of inequality constraints  and $Z=0$  for equality constraints.   Here $A$ is a vector of linear functionals.

Our optimizations over $x$ are in function space,  which we take to be the $L^1$  Banach space w.r.t. Lebesgue measure.  We write $\mathcal{X}$ for the feasible region for the unconstrained problem (which requires $x$ to be positive amongst other things).    Hence the constrained problem can be written as 
\begin{align*}
&\bf{P_Z(\theta)}:\\
  &\rm{Maximize}  &&  OBJ = \alpha\ \text{revenue} + \beta\ \text{welfare} + \gamma\ \text{click yield} \\
    & \rm{subject\; to}  &&  A(x)+z = \theta \\
   &   &&  x \in \mathcal{X} \\
  & \mbox{  over  $x \in L^1,  z \in Z$}.
\end{align*}

We now introduce some technical conditions, used to facilitate proofs: we assume that the allocation function $x: T \rightarrow \mathbb{R}^n$ is a Lebesgue integrable function $\in L^1$,  and we require the constraint feasibility region $\mathcal{X}$ for the allocation to be convex with respect to the underlying function space.   Note that ``natural'' requirements for sponsored search auctions translate to convex regions if we allow for randomized allocations.

The Lagrangian for the constrained problem is
 \begin{equation*}
L(x, \lambda) = G(x) - \lambda^ T (\theta - A(x)-z)
\end{equation*}
and we demonstrate that the constrained problem is strong Lagrangian, meaning strong duality applies,  and hence
 \begin{equation*}
                \inf_{\lambda} \sup_{x \in X}L(x,\lambda)=    L(x^*, \lambda^*) = \sup_{x \in \mathcal{X}} L(x, \lambda^*)=\sup_{x: x \in \mathcal{X} \& A(x)+z=\theta} G(x).
\end{equation*}•

The proof largely mirrors the standard proof for  convex optimization over Euclidean spaces,  but applied to function space.   There are two main parts of the proof:  first, we show that if $\theta$ is in the relative interior of the constrained optimization then the problem is strong Lagrangian.  That is, defining $\Gamma(\theta) = \sup_{x: x \in \mathcal{X} \& A(x)+z=\theta} G(x)$,  we show the following.

\begin{theorem}
\label{thm:strongL}
If $\theta \in \relint(\dom(\Gamma))$, then the optimization problem is strong Lagrangian.
\end{theorem}
The second part is to give the equivalent version of Slater's constraint qualification --- which gives sufficient conditions for the problem to be strong Lagrangian.
\begin{corollary}
  When $Z=\{z\in \bR^{\rho(A)}: \leq 0\} $ is the non-positive cone, the problem $P_{Z}(\theta)$, is strong Lagrangian if there is some $\tilde{x} \in \mathcal{X}$ such that $A(\tilde{x})>\theta$.  If $Z={0}$, the problem  $P_{Z}(\theta)=P(\theta)$ is strong Lagrangian provided there is some $\tilde{x} \in \relint(\mathcal{X})$ such that $A(\tilde{x})=\theta$.
\end{corollary}

Appendix \ref{sec:constraint_proof} goes through the  proofs in detail.

The next theorem is then a direct consequence of the problem being strong Lagrangian, and shows that auctions which satisfy linear constraints are essentially those which have the aforementioned linear objectives.  We give the theorem for inequalities (the theorem for equalities is similarly proved).

\begin{theorem}
\label{thm:constraint}
Suppose $x^*$ maximizes $\alpha_0\ \text{revenue} + \beta_0\ \text{welfare} + \gamma_0\ \text{click yield}$ subject to the following constraints:
\begin{align*}
&\alpha_1\ \text{revenue} + \beta_1\ \text{welfare} + \gamma_1\ \text{click yield} \geq  \theta_1\\
&\alpha_2\ \text{revenue} + \beta_2\ \text{welfare} + \gamma_2\ \text{click yield} \geq \theta_2\\
&\qquad\qquad\vdots\\
&\alpha_r\ \text{revenue} + \beta_r\ \text{welfare} + \gamma_r\ \text{click yield} \geq \theta_r.
\end{align*}
Then there exists some $\alpha^*$, $\beta^*$, and $\gamma^*$ such that $x^*$ maximizes $\alpha^*\ \text{revenue} + \beta^*\ \text{welfare} + \gamma^*\ \text{click yield}$.
\end{theorem}
\begin{proof}
Since a solution $x^*$ exists, the feasible region is non-empty.   If a feasible solution exists with all the inequalities strict, then we know the problem is strong Lagrangian, and hence there are Lagrange multipliers $\lambda^* = [\lambda^*_1, ...,\lambda^*_r]$ such that the Lagrangian $L(x^*,\lambda^*)$  solves the problem --- and the latter can be written as  OBJ with $ \alpha^*=\alpha_0 + \sum_i \alpha_i  \lambda^*_i$, $ \beta^*=\beta_0 + \sum_i \beta_i  \lambda^*_i$, $ \gamma^*=\gamma_0 + \sum_i \gamma_i  \lambda^*_i$.  Note that complementary slackness ensures that $\lambda.(\theta - A(x^*)) = 0$.  If no feasible solution exists except on the boundary of the constraint set,  then essentially there is one solution, whose values determine $\alpha^*, \beta^*, \gamma^*$.
\end{proof}

Each of the stakeholders would like to dominate the objective function (by setting their term in OBJ to 1 and the other terms to zero). We call an auction approximately optimal with level $\delta$  for one of the stakeholders if it achieves at least a $\delta$ fraction of the utility the optimal auction achieves for that stakeholder alone. One consequence of the theorem above is that the Pareto optimal surface is concave, so that there always exists a deterministic auction that achieves any convex combination of optimality levels for all the stakeholders.
\begin{theorem}
The Pareto surface is concave. That is, the surface of points $($revenue, welfare, click yield$)$ that are feasible by some direct revelation auction such that there is no other auction that simultaneously achieves weakly better revenue, welfare and click yield, and also achieves strictly better of one of the tree, is concave.
\end{theorem}
\begin{proof}
Fix some constraint $\theta$.  Then we know that for every $OBJ(\alpha, \beta, \gamma)$ there is a non-vertical supporting hyperplane (the Lagrangian) at the point $\theta$ to the hypograph $\Gamma(\theta)$ which bounds the (constrained) objective function from above. Taking the point-wise infimum of these affine hyperplanes as $\alpha$, $\beta$, and $\gamma$ vary results in a concave surface. 
\end{proof}

\begin{corollary}
Let $\alpha, \beta, \gamma \geq 0$ satisfy $\alpha + \beta + \gamma = 1$. Then there exists a deterministic auction that achieves at least an $\alpha$ fraction of the max revenue, a $\beta$ fraction of the max welfare, and a $\gamma$ fraction of the max click yield.
\end{corollary}

Another class of constrained auctions of interest are those where we limit the number of ads shown. That is, given $J$ search queries over  various search terms we may wish to show no more than $k$ ads  in expectation.
This is another natural way to capture the utility of searchers, who presumably would generally prefer to see fewer, better-targeted ads.
We now investigate such ad-limited auctions, making the simplifying assumption that all $J$ auctions are single item auctions.
In contrast to previous work which has suggested that a per-impression reserve is a poor tool for increasing revenue or trading-off between revenue and other objectives~\cite{roberts13,thompson13}, this analysis shows that the use of such a reserve is the \emph{optimal} way of controlling the number of ads shown.

\begin{theorem}
\label{thm:per-impression}
Suppose we have $J$ single slot search terms where search term $j$ appears with probability $q_j$ and has a slot effect of $s_j$ for its lone slot. Moreover, suppose we wish to show no more than $\theta \leq 1$ ads per search (in expectation) and when auctioning off search term $j$, we wish to maximize $\alpha_j\ \text{revenue} + \beta_j\ \text{welfare} + \gamma_j\ \text{click yield}$. Then the optimal solution introduces a per-impression reserve price in all $J$ auctions. Specifically, for each search term $j$ we will separately wish to solve the following:
\begin{align*}
\max_{x_j} \int_{T_j} \left(\sum_{i = 1}^{n_j} \left(w_{j, i}\psi_{j, i}(t_{j, i}) - \lambda/s_j\right) x_{j, i}(t_j)\right) f_j(t_j) dt_j
\end{align*}
where quantities indexed by $j$ refer to the $j$th search term's auction (e.g. $\psi_{j, i}$ is the $\psi$ function for search term $j$'s bidder $i$), $\lambda$ is some common value to all search terms, and $\alpha_j, \beta_j, \gamma_j \geq 0$ with $\alpha_j+\beta_j+\gamma_j>0$  for each $j$.
\end{theorem}
\begin{proof}
We first note that the expression is monotone, as the sum of monotone components. Now via a similar analysis to the theorem \ref{thm:opt} it is not difficult to see that our problem is equivalent to:
\begin{align*}
\max_{x_1, x_2, ..., x_J}& \sum_{j = 1}^J q_j \int_{T_j} \sum_{i = 1}^{n_j} w_{j, i} \psi_{j, i}(t_{j, i}) x_{j, i}(t_j) f_j(t_j) dt_j\\
s.t.& \sum_{j = 1}^J q_j \int_{T_j} \sum_{i = 1}^{n_j} \frac{x_{j, i}(t_j)}{s_j} f_j(t_j) dt_j = \theta.
\end{align*}
Using a Lagrange multiplier $\lambda$  gives a dual problem for fixed $\lambda$, $L(\lambda,x^*)$
\begin{align*}
\max_{x_1, x_2, ..., x_J} \lambda\theta + \sum_{j = 1}^J q_j \int_{T_j} \left(\sum_{i = 1}^{n_j} \left(w_{j, i}\psi_{j, i}(t_{j, i}) - \lambda/s_j\right) x_{j, i}(t_j)\right) f_j(t_j) dt_j.
\end{align*}
For a fixed value of $\lambda$ not only does the $\lambda\theta$ in the above maximization become irrelevant, but also the terms in the sum indexed by $j$ are independent of each other. Thus, the maximum of the sum can be separated into the sum of the maximums. The resulting solution (for fixed $\lambda$) is feasible;   taking the minimum over $\lambda$  of $L(\lambda,x^*)$  also produces a solution which is feasible, and hence by Lagrangian sufficiency $L(\lambda^*, x^*)$ solves the constrained problem.  (c.f. Lemma \ref{lem:LagSuff}).
\end{proof}

\section{Experiments}

Roberts et al.~\citeyear{roberts13} empirically considered a number of ways of making trade-offs between revenue and clicks / welfare and found that a form of per-click reserve they proposed (independently proposed by Thompson and Leyton-Brown~\citeyear{thompson13}) lead to the best results.  In Theorem~\ref{thm:opt}, we derived the optimal way of making such trade-offs, by optimizing a weighted combination of virtual valuations, clicks, and welfare. The proposed form can be thought of as a linear approximation to this, providing a theoretical explanation for this observation.  Another observation from this prior work is that, despite their popularity, per-impression reserve prices tended to enable poor tradeoffs.  However, Theorem~\ref{thm:per-impression} suggests they are the right tool for controlling the number of ads shown, as they correspond to the shadow price of a constraint on the number of ads.  In this section, we revisit the experiments from \cite{roberts13} to validate this empirically.

Our experiments look at four metrics: the number of ads shown (impression yield), the revenue, the (advertiser) welfare, and the number of clicks generated (click yield).  In the short term, the auctioneer cares mostly about the revenue metric. Welfare reflects the total value created, and is a focus of the auctioneer in the long term, as it determines the ability of the platform to attract participants (who may have alternative platforms available). Click yield can be thought as a proxy for the value created for the users who are clicking on (presumably) useful ads. Impression yield can be thought of as another proxy for the value created for users, who presumably do not wish to see useless ads.

We consider several ranking algorithms in our analysis: the method Roberts et al.~\citeyear{roberts13} and Thompson and Leyton-Brown~\citeyear{thompson13} used to approximate the revenue optimal auction (denoted by $(b_i - r) w_i$, where $r$ is a per-click reserve), the standard ranking algorithm coupled with a per-click reserve $r$ (denoted by $b_iw_i / r$) and with a per-impression reserve $\rho$ (denoted by $b_iw_i / \rho$), and a two parameter ranking algorithm that combines the first and the third (denoted by $(b_i-r)w_i / \rho$). Since this ranking algorithm is characterized by two parameters, the operating points form a region represented in the following figures by a shaded area.

We begin with a simple example characterized by eight advertisers bidding for three slots. Advertisers have i.i.d. types $(t_i, w_i)$ where $t_i$ and $w_i$ are independent and uniformly distributed on $[0, 1]$. Figure~\ref{fig:uniform} shows how each of the metrics changes with the allowed number of ads shown, using different auction designs. 

%

\begin{figure}[ht]
\hspace{-0.6cm}
\begin{minipage}[b]{0.3\linewidth}
\centering
\includegraphics[scale=.355]{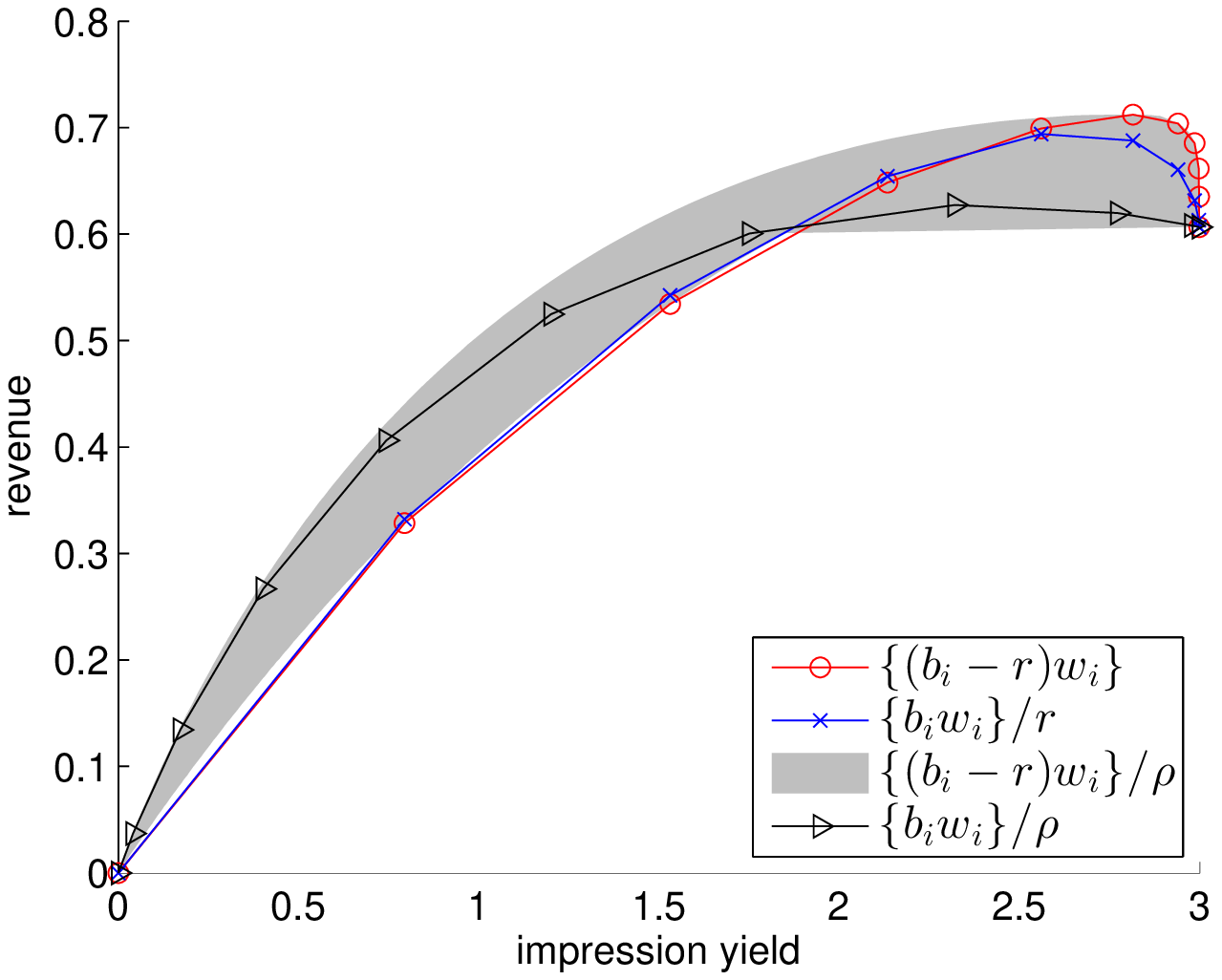}
\scriptsize{(a)}
\label{fig:revenue-uniform}
\end{minipage}
\hspace{0.5cm}
\begin{minipage}[b]{0.3\linewidth}
\centering
\includegraphics[scale=.355]{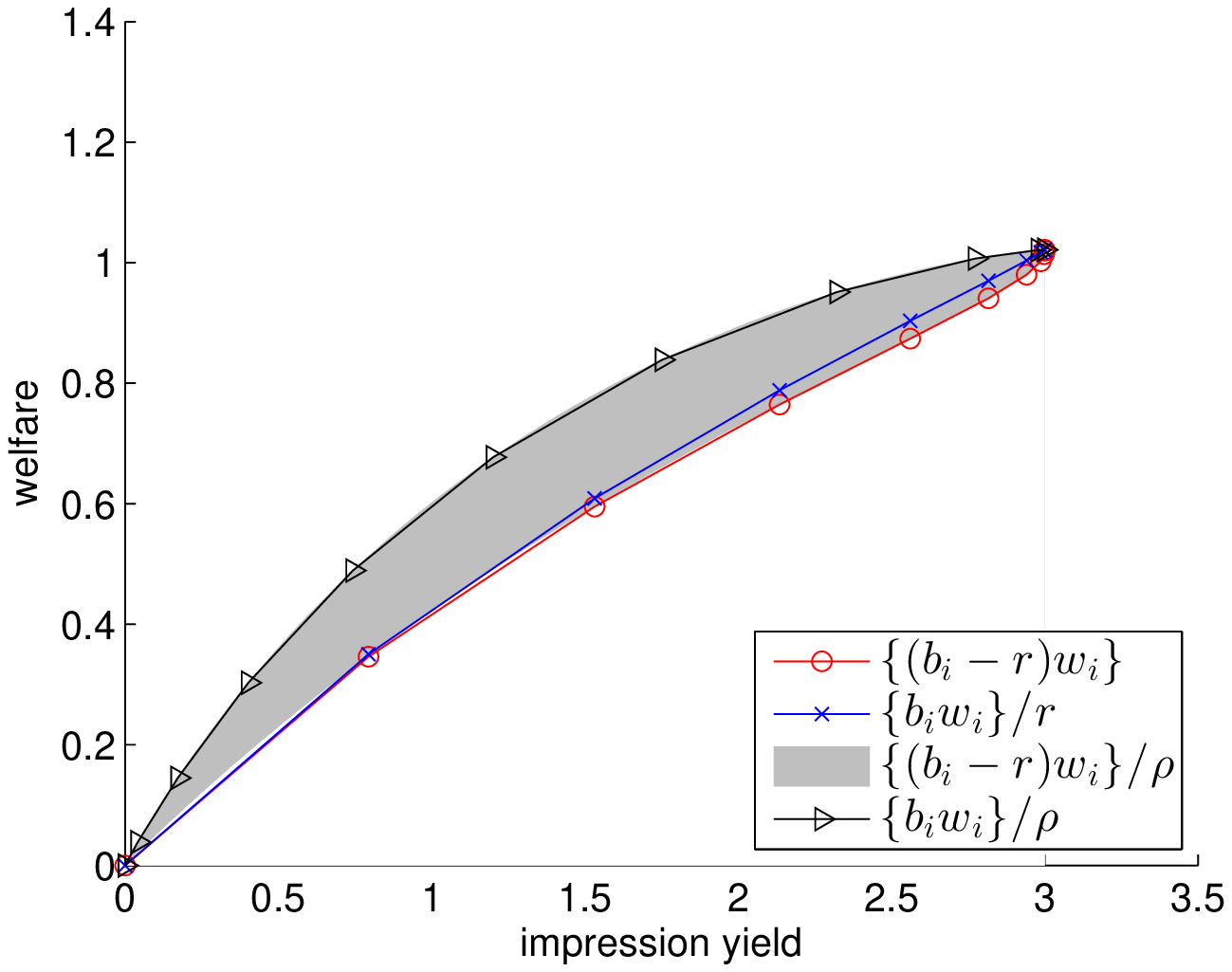}
\scriptsize{(b)}
\label{fig:social-uniform}
\end{minipage}
\hspace{0.5cm}
\begin{minipage}[b]{0.3\linewidth}
\centering
\includegraphics[scale=.355]{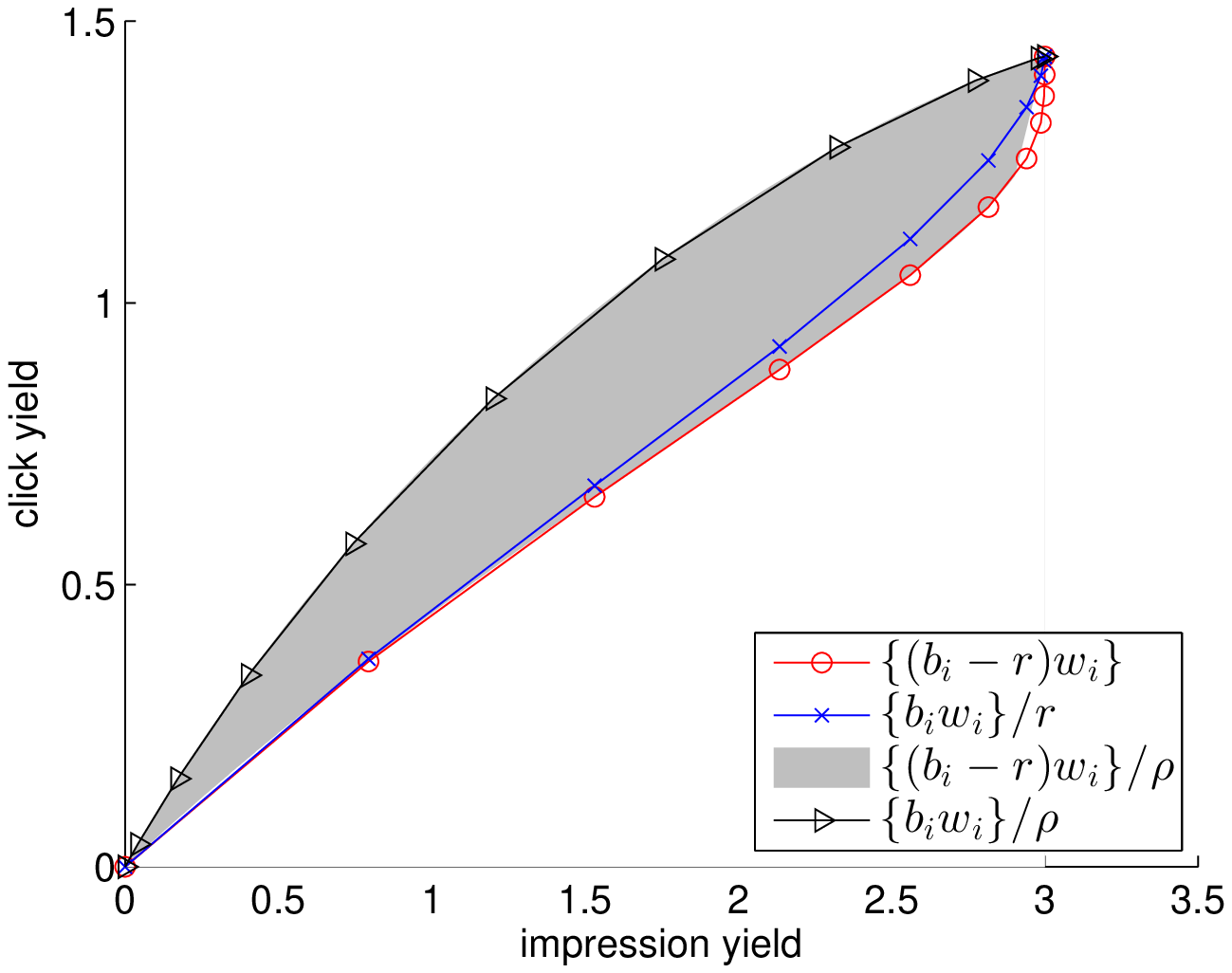}
\scriptsize{(c)}
\label{fig:click-uniform}
\end{minipage}
\caption{Feasible impression yield -- revenue (a), impression yield -- welfare (c), and impression yield -- click yield (c) operating points in a simple setting. \label{fig:uniform}}
\end{figure}


The figure illustrates Theorem~\ref{thm:per-impression}: given an expected number of ads shown, the optimal solution is provided by a ranking algorithm with a reserve score. In particular, the welfare and the click yield are always maximized by the standard ranking algorithm with a reserve score.  The revenue is always maximized by the two-parameter algorithm.

We also evaluate the ranking algorithms in a more realistic setting considered by Lahaie and Pennock~\citeyear{lahaie07}), which they selected by fitting Yahoo! data from a particular query. In this setting, bidder valuations are correlated with relevance, and have a lognormal distribution (which does not yield regular virtual valuations). The results for this setting are shown in
Figure~\ref{fig:log}, which indicates that the observations made for the simpler setting hold in this case as well.

\begin{figure}[ht]
\hspace{-0.6cm}
\begin{minipage}[b]{0.3\linewidth}
\centering
\includegraphics[scale=.355]{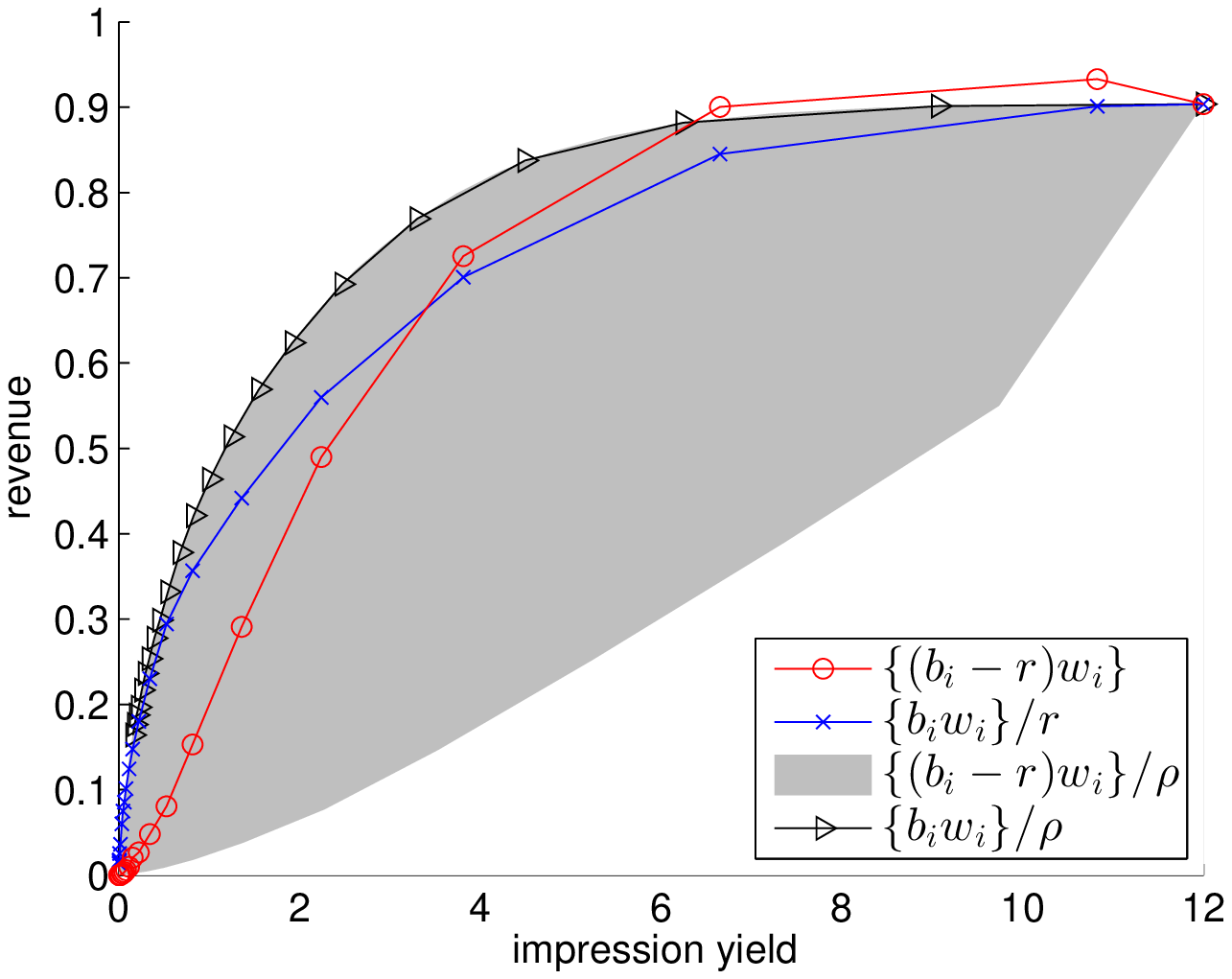}
\scriptsize{(a)}
\label{fig:revenue-log}
\end{minipage}
\hspace{0.5cm}
\begin{minipage}[b]{0.3\linewidth}
\centering
\includegraphics[scale=.355]{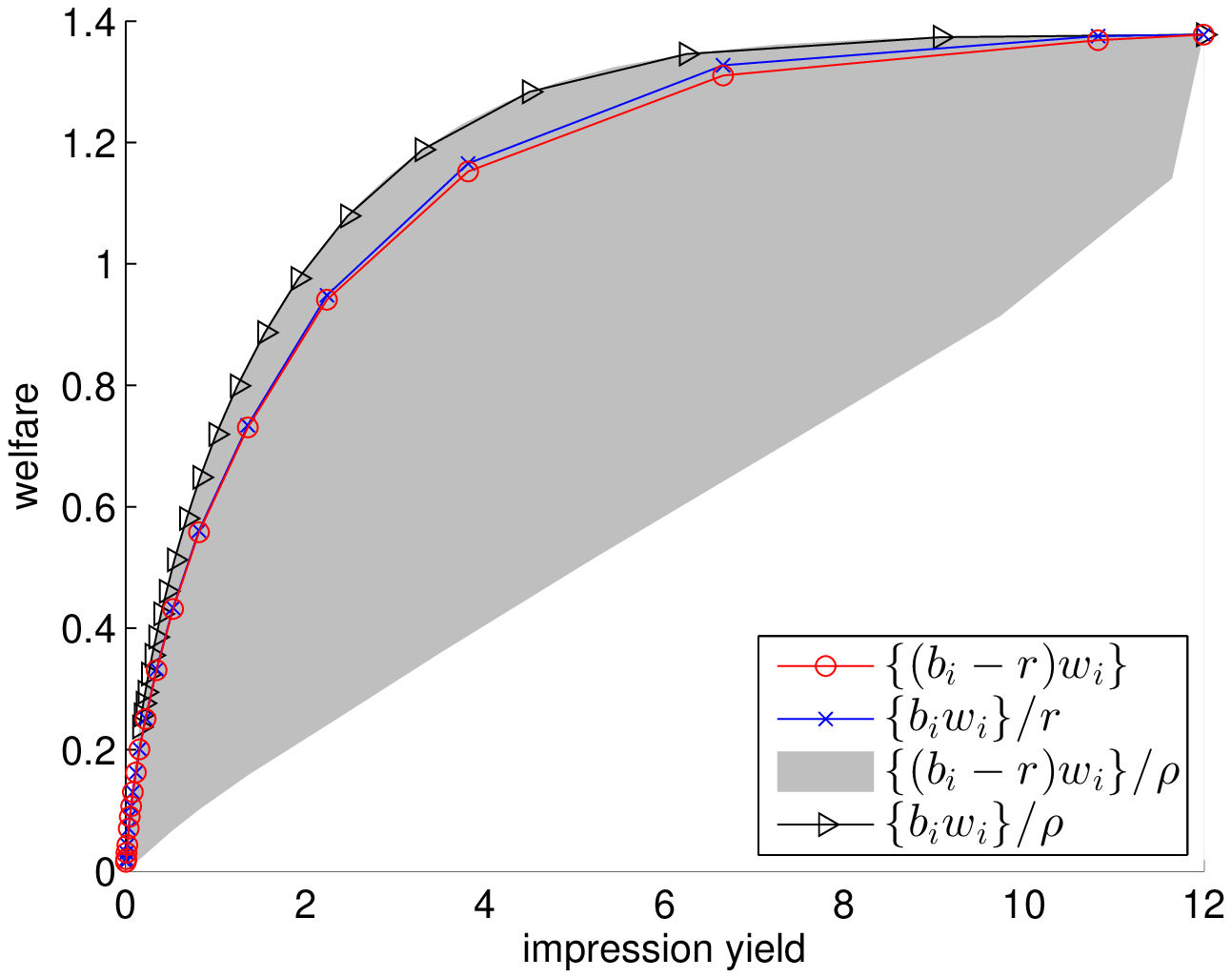}
\scriptsize{(b)}
\label{fig:social-log}
\end{minipage}
\hspace{0.5cm}
\begin{minipage}[b]{0.3\linewidth}
\centering
\includegraphics[scale=.355]{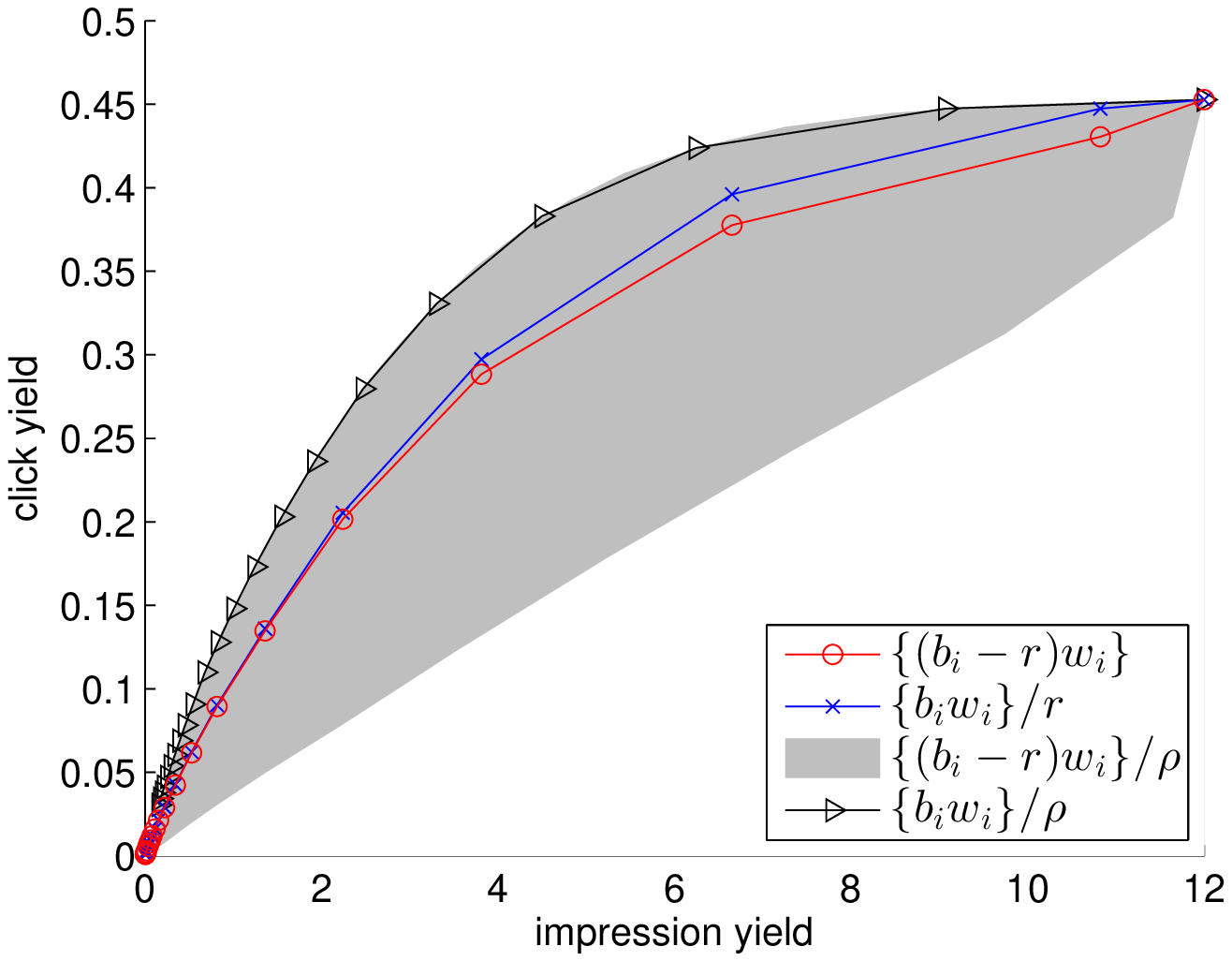}
\scriptsize{(c)}
\label{fig:click-log}
\end{minipage}
\caption{Feasible impression yield -- revenue (a), impression yield -- welfare (c), and impression yield -- click yield (c) operating points  in Lahaie and Pennock's setting.
\label{fig:log}}
\end{figure}

The above results are based on the assumption that bidders are in equilibrium. In reality, if parameters are changed, advertisers may take some time to reach equilibrium.
Therefore, we want to investigate what happens in the short term, when the ranking method is changed but advertisers do not react.

As we are not interested in the equilibrium in this case, we can simply examine the performance of different ranking algorithms on historical data. This dataset has many realistic features which were not captured by the previously analysis, such as changing bidders, matching of bids to multiple queries, and stochastic quality scores.

The data we used for our simulations (used also in~\cite{roberts13}) is historical data from Microsoft Bing for a keyword with over 500 bidders, which we selected as representative of a ''thick'' market (Figure~\ref{fig:thick-market}), and for a keyword with fewer than 10 bidders, as representative of a ''thin'' market (Figure~\ref{fig:thin-market}). The data was normalized, but the exact values are not relevant for our purpose. Both figures show the effect of the number of ads displayed on the revenue when changing from the standard ranking algorithm with a reserve price to the ranking algorithm proposed by Roberts et al.~\citeyear{roberts13} and the standard ranking algorithm with a reserve score.

In the thick market, no ranking algorithm always outperforms all the others. However, though the three algorithms show a similar trend, the standard algorithm with reserve score provides a higher revenue when only few ads are displayed. In the thin market, the standard ranking algorithm with reserve score exhibits a different trend than the other algorithms (which are similar to each other in their behavior). In particular, the revenue provided by the former algorithm remains constant for almost all the number of ads shown. The non-smoothed trend all the ranking algorithms show is explained by the advertisers' bids. In the dataset we use, some bids are more frequent than others (see~\cite{roberts13}) and this creates thresholds for the reserve price/score: when the reserve increases and moves from below a threshold to above it, the impression yield suddenly reduces. Interestingly, in both markets  the standard ranking algorithm with reserve score performs better than the others for the same range of ads shown. 

\begin{figure}[!htb]
\begin{minipage}[b]{0.45\linewidth}
\centering
\includegraphics[scale=.46]{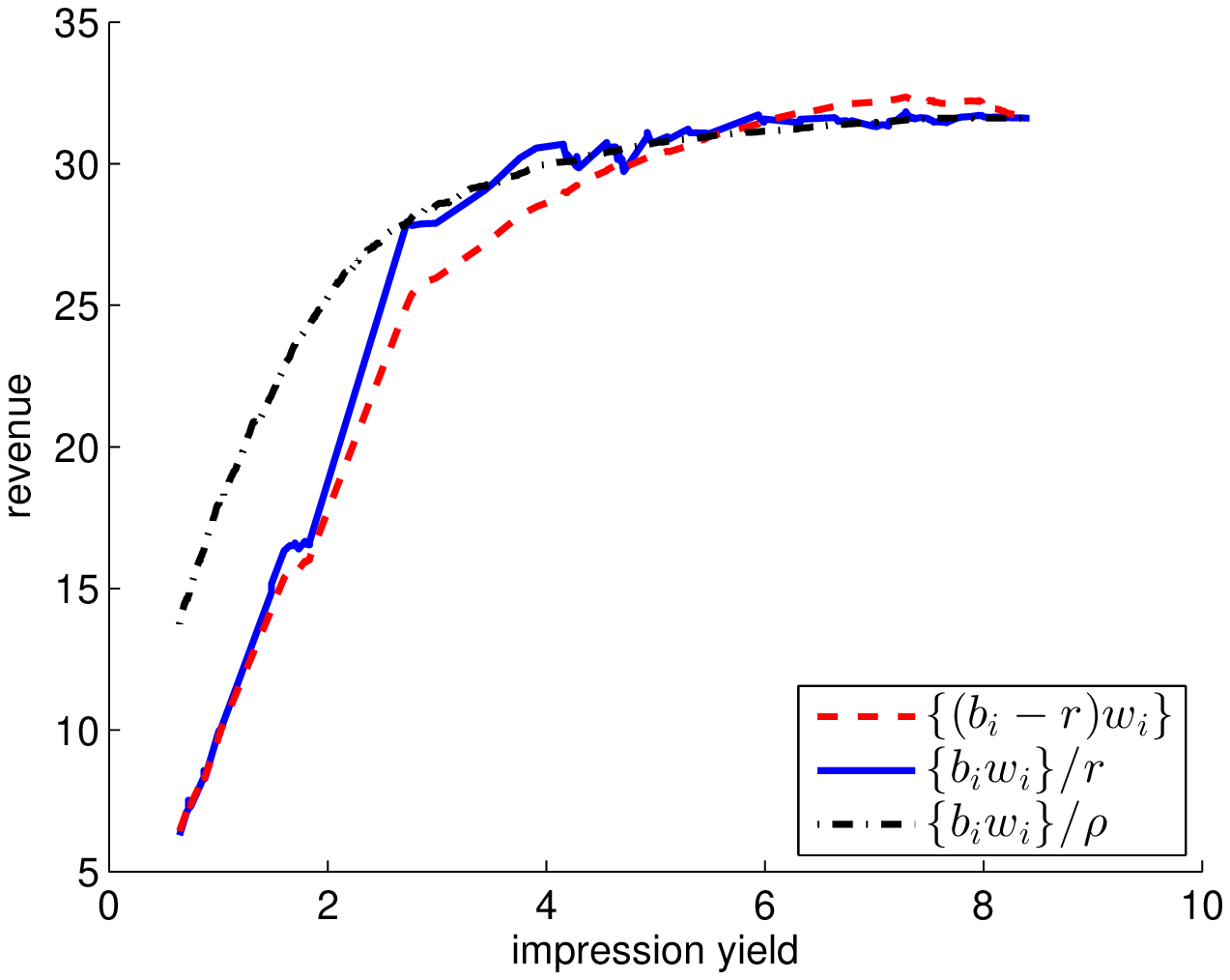}
\caption{Revenue comparison in thick market auction replays.}
\label{fig:thick-market}
\end{minipage}
\hspace{0.5cm}
\begin{minipage}[b]{0.45\linewidth}
\centering
\includegraphics[scale=.46]{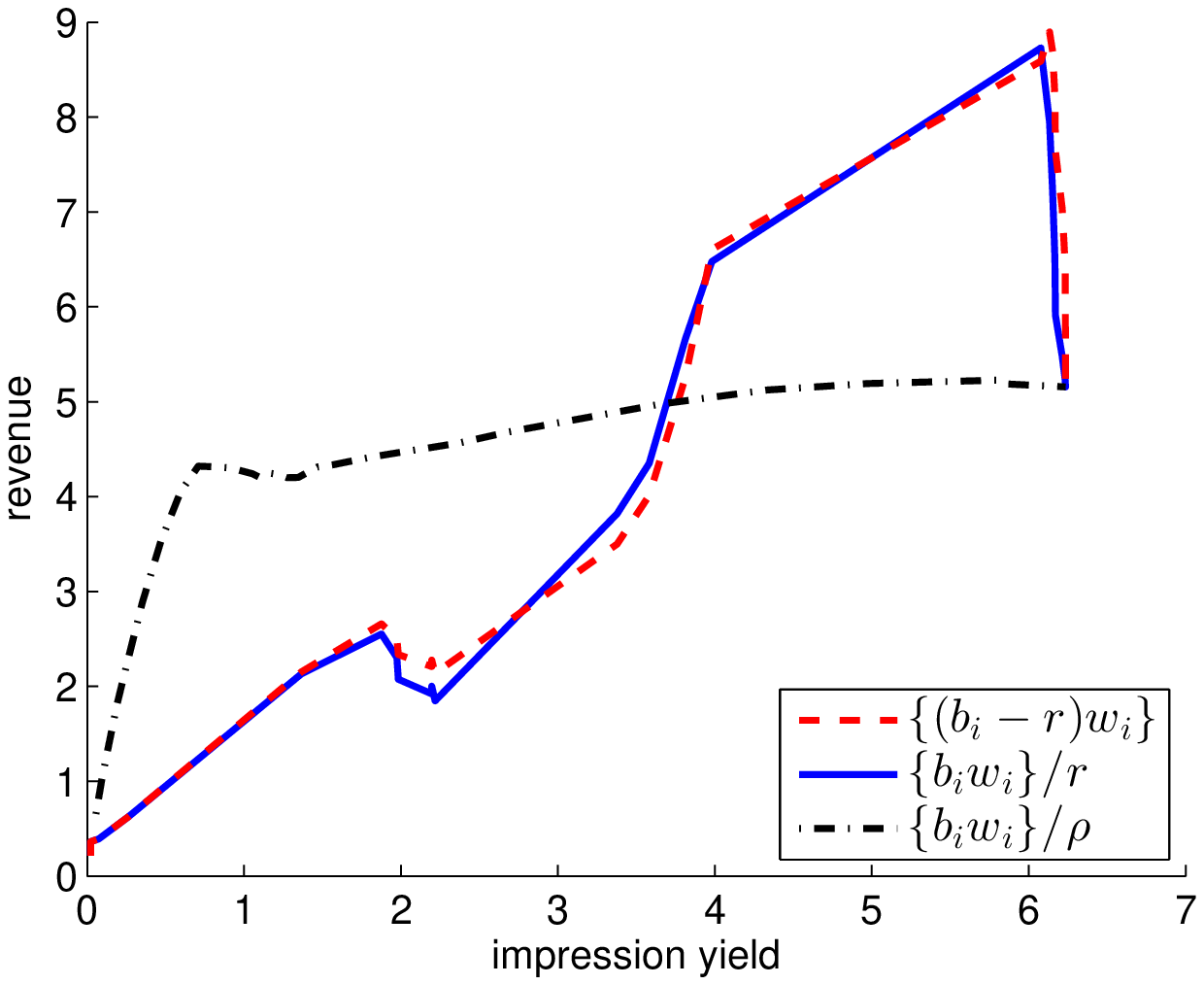}
\caption{Revenue comparison in thin market auction replays.}
\label{fig:thin-market}
\end{minipage}
\end{figure}

\begin{figure}[!htb]
\begin{minipage}[b]{0.45\linewidth}
\centering
\includegraphics[scale=.46]{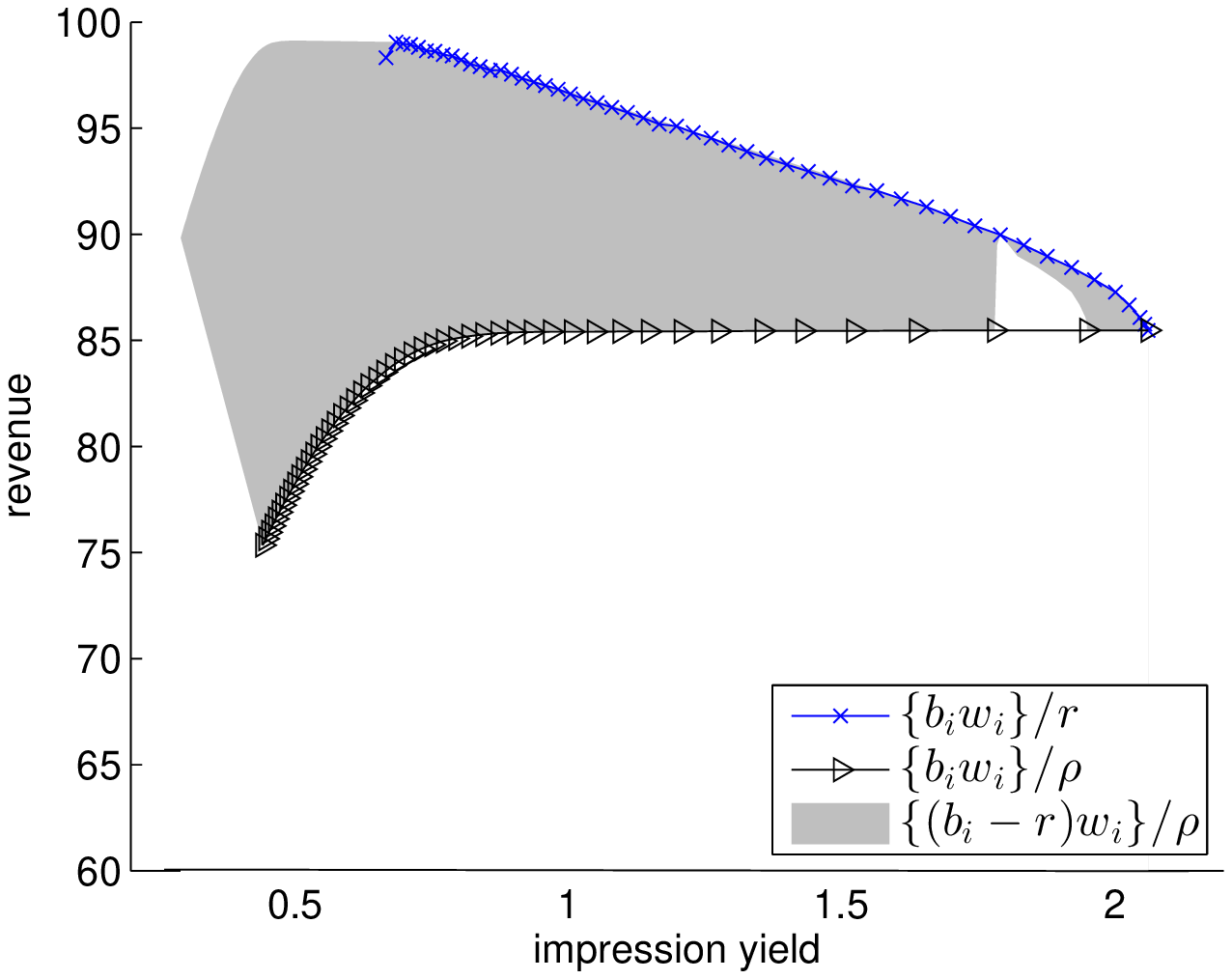}
\caption{Revenue trade-offs on real data}
\label{fig:rev-real}
\end{minipage}
\hspace{0.5cm}
\begin{minipage}[b]{0.45\linewidth}
\centering
\includegraphics[scale=.46]{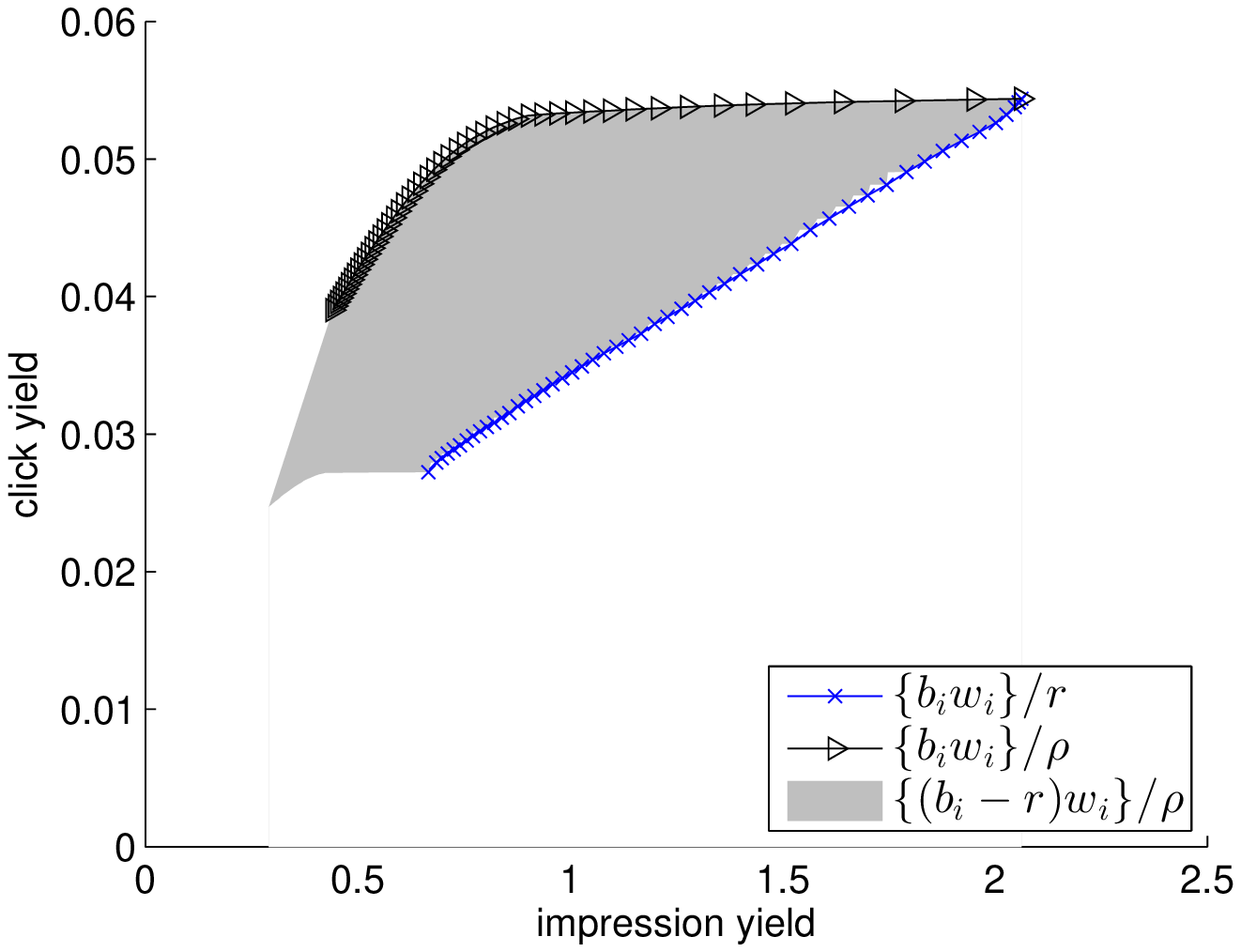}
\caption{Click trade-offs on real data}
\label{fig:click-real}
\end{minipage}
\end{figure}

Finally, on a sample from a week's worth of data across all keywords on Bing, Figures~\ref{fig:rev-real} and~\ref{fig:click-real} show exactly opposite results when trading off between clicks and impressions versus revenue and impressions.  For clicks, consistent with the theory, the optimal parameter to use is a per-impression reserve $\rho$.  For revenue, as raising the per-click reserve always increased revenue for most of the range we explored, using that as the sole parameter is optimal for most of the range.  Only once increasing it further begins to reduce revenue does it become optimal to add in a per-impression reserve.  Again, this is consistent with the theory, despite this data not reflecting equilibrium behaviour.  



\section{GSP Auctions with Classes and Templates}

We now turn to examine a more elaborate constrained optimization problem. Our motivation comes from changes in the world of sponsored search auctions, where recently several platforms have allowed enriching ads with features such as larger formats, reviews, maps, or phone numbers. A richer ad augmented with images or video is more likely to be clicked, but uses up more screen space. As the same space on the screen can be used in different ways or formats, the search engine faces the problem of determining which of many possible arrangements of advertising, or {\em templates} to use. Deng et al.~\citeyear{Deng2010} make an initial step to solve this problem by proposing and studying the multi-slots pricing scenario.

As a guiding example for our work, consider the setting studied by Goel and Khani~\citeyear{goel14}, where a search result page can either display one image ad or $k$ text ads. 
To analyse trade-offs in this type of setting, we require the following two extensions to the classical framework.
\begin{itemize}
\item There may be several \emph{classes} of ads/bidders that are disparate in nature: the bidders of a class bid only for ads of their own class. Importantly, we assume throughout this section that the sets of bidders of two classes are {\em disjoint}. 
In our example, one class consists of the advertisers interested only in text ads, while the other class consists of those interested only in image-rich ads.
\item There may be several mutually-exclusive \emph{templates}, each of which describe an entire slot layout. 
For example, one template allows a single image-rich ad (and no text ads) and another template allows $k$ text ads (but no image-rich ads).
\end{itemize}

With these extensions to the classical framework we will require the following updates to our notation and terminology.
\begin{itemize}
\item $x_{i, j} : A \rightarrow \mathbb{R}^n$: the allocation function of bidder $i$ on template $j$. Intuitively, this represents the slot effect given to bidder $i$ if the search engine decides to use template $j$. We therefore have the following constraint. Suppose template $j$ has slot effects $s_{j, 1}, s_{j, 2}, ..., s_ {j, n}$ (we can assume there is a slot for every player by setting $s_{j, i} = 0$ as needed). For every $\ell \in \{1, 2, ..., n\}$, let $\tilde{x}_{i, j, \ell} : A \rightarrow [0, 1]$ be functions such that for every $b$, we have $\sum_i \tilde{x}_{i, j, \ell}(b), \sum_\ell \tilde{x}_{i, j, \ell}(b) \leq 1$. Intuitively $\tilde{x}_{i, j, \ell}$ represents the probability that bidder $i$ receives slot $\ell$ (in template $j$). Then the $x_{i, j}$ are valid allocation functions only if there exists some $\tilde{x}_{i, j, \ell}$ such that for every $i, j,$ and $b$, we have $x_{i, j}(b) = \sum_\ell s_\ell \tilde{x}_{i, j, \ell}(b)$. If they satisfy this constraint and do not assign slots from one class to a bidder of another, then the $x_{i, j}$ are valid.
\item $x_i : A \rightarrow \mathbb{R}^n$: the actual allocation rule chosen. For every set of bids $b$, there must exist some template $j$ such that for all $i$ we have $x_i = x_{i, j}$. Note that the order of quantifiers allows the chosen template to depend on the bids.
\end{itemize}

Our first observation is that, for a truthful auction, this complexity does not introduce any significant issues. We are able to select the optimal template for our objectives, and as the domain remains a single-parameter one we are able to charge appropriate payments to induce truthful bidding. 
\begin{theorem}
\label{thm:optimal}
A truthful auction that maximizes $OBJ$ is described by the allocation rule $x_{i, j}$:
\begin{align*}
\max_j \max_x \sum_{i = 1}^n w_i \psi(b_i) x_{i, j}(b). 
\end{align*}
\end{theorem}
\begin{proof}
This allocation rule selects the template that has the highest value of $OBJ$ when allocated optimally (correctly optimizing revenue because we are still in a single parameter domain), so it only remains to check monotonicity to verify truthfulness.
Suppose for purposes of contradiction it is not. Then there exists a bidder $i$ such that when he raises his bid from $b_i$ to $b_i'$ the slot-effect assigned to him is decreased in value. Explicitly, we have that if $s_j$ is the allocated slot effect for $j$ when $i$ reports a bid of $b_i$, and $s_j'$ when $b_i'$ is reported, then $s_i' < s_i$. Now as the $s_j$ are chosen optimally for $b$, we have that:
\begin{align*}
&\sum_{j = 1}^n w_j \psi(b_j) s_j \geq \sum_{j = 1}^n w_j \psi(b_j) s_j'\\
&\Rightarrow \psi(b_i) \geq \frac{\sum_{j = 1, j \neq i} w_j \psi(b_j) (s_j' - s_j)}{w_i(s_i - s_i')}.
\end{align*}
Similarly, as the $s_j'$ are chosen optimally for $(b_i', b_{-i})$, we have that:
\begin{align*}
&w_i \psi(b_i') s_i' + \sum_{j = 1, j \neq i}^n w_j \psi(b_j) s_j' \geq w_i \psi(b_i') s_i + \sum_{j = 1, j \neq i}^n w_j \psi(b_j) s_j\\
&\Rightarrow \psi(b_i') \leq \frac{\sum_{j = 1, j \neq i} w_j \psi(b_j) (s_j' - s_j)}{w_i(s_i - s_i')}.
\end{align*}
Combining these two results, we find that $\psi(b_i) \geq \psi(b_i')$, and as $\psi$ is an increasing function, this is a contradiction.
\end{proof}

\subsection{Extensions of GSP}

We have shown that our earlier results generalize for {\em truthful} auctions. Unfortunately, as previously discussed, search engines such as Bing and Google use GSP instead. Previously, this was not a significant issue (except possibly for exact revenue maximization) due to the equivalence results of Varian~\citeyear{varian07}, Edelman et al.~\citeyear{eos07}, and more generally Roberts et al.~\citeyear{roberts13}. However, due to the combinatorial nature of the image/text auction, and templates more broadly, the distinctions between truthful and non-truthful auctions become critical. 
Indeed, it is not even clear what the ``right'' way is to generalize GSP to our setting, and as we shall see, simple candidates have problems even when solely optimizing for advertiser welfare.

An underlying principle in standard GSP auctions is that ``bidders pay the minimum required to keep their slot''. We show that several natural extensions of GSP that maintain this principle are provably {\em not equivalent} to the truthful mechanism and have poor equilibrium properties.

Before we begin, an extension of the notion of the SNE is in order.
\begin{definition}
A \emph{symmetric Nash equilibrium} is a set of bids $b$ that satisfy the conditions:
\begin{itemize}
\item for all $i$ and $b_i'$ we have: $x_i(b_i', b_{-i})(t_i - p_i(b_i', b_{-i})) \leq x_i(b)(t_i - p_i(b))$
\item for all $i$ and $j$ where $i$ and $j$ are bidding on the same class of ads, we have: $x_j(b)\left(t_i - \frac{w_j}{w_i}p_j(b)\right) \leq x_i(b)(t_i - p_i(b))$
\end{itemize}
\end{definition}
The first of the two conditions is what constitutes a Nash equilibrium and the second defines the symmetric/envy-free nature of the equilibrium. Note that the second does not encompass the first, unlike in the classical single-template scenario.

\subsubsection{Template-Considerate GSP}
We begin with perhaps the most straightforward extension to multi-template settings.
\begin{definition}
The \emph{template-considerate} GSP payment rule is such that the payment is the value of the minimum bid to ensure the bidder keeps his current slot. 
That is, the payment is the maximum of the minimum bid needed to be the next highest bidder's bid in the same class, and the minimum bid required to ensure the winning template remains optimal.
\end{definition}

Unfortunately, this payment rule suffers from a few critical issues. First, we show that even in the very restricted image/text setting, an SNE may fail to select the same outcome as the truthful mechanism.  Thus, the auction fails to optimize the desired objective.

\begin{theorem}
Suppose we use the template-considerate GSP payment rule. Then we may not implement the truthful outcome.
\end{theorem}
\begin{proof}
Suppose we are maximizing welfare and we have the following setup with all the $w_i = 1$.
\begin{itemize}
\item Bidders:
\begin{itemize}
\item Four text ad bidders with true values $100$, $50$, $25$, and $10$.
\item Two image ad bidders with true values $120$ and $110$.
\end{itemize}
\item Templates:
\begin{itemize}
\item Template $1$: three text ads with slot effects $1$, $1 - \epsilon$, and $1 - 2\epsilon$.
\item Template $2$: a single image ad with slot effect $1$.
\end{itemize}
\end{itemize}
In this given setup one problematic SNE is when the two image ad bidders bid their true value, while the text ad bidders bid zero. On the other hand, the truthful outcome clearly selects the first template.
\end{proof}
Note that this example relies on losing bidders whose class is not shown bidding 0.  If instead we add the requirement that such bidders bid their true value, it is possible to construct examples where SNE do not exist.

Second, In more general settings, an SNE may not even exist (even without such an added requirement).  
Essentially, the issue is that the constraints on bids imposed by the requirement for a lack of envy on a given template may force the bids to be inconsistent with that template being chosen.  
This example does rely on bidders being {\em conservative}~\cite{paesleme09} and not bidding above their true value (which is a dominated strategy).

\begin{theorem}
\label{thm:tcgsp_nonexistence}
Suppose we use the template-considerate GSP payment rule and bidders do not bid higher than their true value. 
Then there may not exist any SNE.
\end{theorem}
\begin{proof}
Suppose we are maximizing welfare and we have the following setup where $\epsilon \approx 0$ with all the $w_i = 1$.
\begin{itemize}
\item Classes $A$ and $B$ (that are entirely symmetrical).
\item Bidders:
\begin{itemize}
\item Four class $A$ bidders with true values $350$, $300$, $200$, and $100$.
\item Four class $B$ bidders with true values $350$, $300$, $200$, and $100$.
\end{itemize}
\item Templates:
\begin{itemize}
\item Template $1$: three class $A$ slots of effects $1$, $1 - \epsilon$, and $1 - 2\epsilon$, and three class $B$ slots of effects $\epsilon$, $\epsilon^2$, and $\epsilon^3$
\item Template $2$: three class $B$ slots of effects $1$, $1 - \epsilon$, and $1 - 2\epsilon$, and three class $A$ slots of effects $\epsilon$, $\epsilon^2$, and $\epsilon^3$
\end{itemize}
\end{itemize}

Suppose that the first template is the winning template in an SNE. We claim the four class $A$ bidders can bid at most (approximately) $350$, $200$, $200$, and $100$ respectively. All but the second bidder are bidding their true value so only the second bidder requires explanation. They cannot bid more than (approximately) $200$ as otherwise the first bidder will envy the second bidder's slot for its price. This is due to the near identical value of the slots. Thus, the objective value of the first template is at most (approximately) $350 + 200 + 200 = 750$. Alternatively, we claim that the four class $B$ bidders must bid at least (approximately) $300$, $300$, $200$, and $100$ respectively. The first bidder must bid at least the second bidder's bid and the last bidder must bid their true value as they will not receive a slot. The second bidder must bid approximately their true value as the first slot is vastly superior to the second and so they must bid near their value to ensure that they are not envious of the first bidder. Similarly, due to the vast superiority of the second slot to the third the third bidder must bid approximately their true value. Thus, the objective value of the second template is at least (approximately) $300 + 300 + 200 = 800$. As this is larger than the first template's worth this is a contradiction. The second template cannot be the winning template in an SNE via similar analysis.
\end{proof}
While our example uses ties in valuations and slot effects for ease of exposition, the values above can be slightly perturbed so that they are not tied and the analysis is not fundamentally affected. Such issues are therefore not zero probability events.

Third, even if an SNE exists, the objective value may be arbitrarily worse in comparison to the truthful auction.
\begin{theorem}
\label{thm:tcgsp_unoptimal}
Suppose bidders do not bid higher than their true values. Then there exists an SNE of an auction governed by the template-considerate GSP payment rule that may be arbitrarily worse than the worst NE of the corresponding truthful auction. That is, if $OBJ_{GSP}$ is the worst $OBJ$ value in an SNE for the template-considerate GSP auction, and $OBJ_{truthful}$ is the worst $OBJ$ value in a NE for the corresponding truthful auction, then $OBJ_{GSP}/OBJ_{truthful}$ can be arbitrarily close to zero (with a sufficiently large number of bidders and slots).
\end{theorem}
We defer the proof to Appendix~\ref{sec:tcgsp_unoptimal_proof}.

Aside from these crucial existence and optimality issues, we note that the payment rule can be counter-intuitive to the bidders. 
For example, it is possible that by increasing one's bid, a bidder does not affect his own slot allocation, but decreases the payment required of their competitors (who are being priced by the minimum bid they need to make to maintain the template).

\subsubsection{Template-Indifferent GSP}
With the failure of the template-considerate GSP rule, an alternate extension we can consider is to simply ignore the existence of other templates for pricing purposes.
\begin{definition}
The \emph{template-indifferent} GSP payment rule is one where after the template is selected, GSP is performed as if the winning template were the only one.
\end{definition}
An immediately obvious issue with this payment rule is that the highest bidder of every class has incentive to raise his bid arbitrarily high. This can be rectified by assuming for the purposes of template selection the highest bidder's bid is equivalent to the second highest bid, however this does undermine the optimality of the allocation.  Even with this fix, the three problems identified with template-considerate GSP remain.  As essentially the same examples work, we omit them. We have also explored other extensions to GSP, but have been unable to find one with desirable equilibrium properties in the general case.

%
%
%

\subsection{Positive Results For Restricted Cases}

The lack of a natural generalization of GSP with desirable properties is disheartening. Fortunately, with restrictions on the set of templates we can get somewhat more favourable results. The most trivial restriction is to a single template, where the classical results apply.  
Only slightly less trivial is the case where one class has the same set of slots in every template.  Then, at least for this class, we are effectively again in the single-template setting.

More interestingly, we can return to our guiding example of the image/text setting explored by Goel and Khani~\citeyear{goel14}.  In this setting, VCG is known to have the undesirable property that adding bidders can reduce revenue.  They seek to rectify this by designing a truthful {\em revenue monotone} mechanism which they call MITA, consisting of a monotone allocation rule and the payments to make it truthful.  Interestingly, while MITA does not optimize $OBJ$, we are able to show that the MITA allocation with template-considerate GSP payments does have an SNE that implements the truthful MITA outcome.
\begin{definition}
Let $\psi_i(t_i)$ be the objective value of the text ad with the $i$th highest value (using a linear approximation to $\phi$ if it is non-linear), $\psi_I(t_I)$ be the value of the best image ad, and $s_{I,1}$ and $s_{T,1} \ldots s_{T,k}$ be the slot effects.
Define $C = \{ 1 \leq j \leq k ~|~ \psi_j(t_j) \sum_{i = 1}^j s_{T,i} \geq \psi_I(t_I)s_{I,1} \}$.
The MITA allocation is the first $\max C$ text ads if $C$ is nonempty and the image ad otherwise%
\footnote{Goel and Khani give their allocation rule for the special case where slots are identical, advertisers have identical qualities, and the objective is to maximize welfare.  We have generalized their allocation rule to our setting.}.
\end{definition}
\begin{theorem}
The MITA allocation with template-considerate GSP payments has an SNE that implements the truthful outcome.
\end{theorem}
\begin{proof}
First, suppose the image ad is winning in the truthful outcome.  We construct the SNE by having all ads bid their true value.  
This clearly implements the truthful MITA outcome because the payment rule and bids are exactly the same as in truthful MITA.  
No text ad wishes to deviate, because their payment for a bid high enough to be shown is the same as it would be in the truthful auction.

Now consider the case where the text ads win in the truthful outcome.  Again, the image ad bidders bid their true value.  Let $j^*$ be the number of ads shown in the truthful MITA outcome.  The text ads make the bids they would make in the lowest SNE of the auction where the text ad template is the only template but there is a per-impression reserve of $\psi_I(t_I)s_{I,1}/\left(\sum_{i = 1}^{j^*} s_{T,i}\right)$ (such an SNE exists by a result from~\cite{roberts13}).  The truthful outcome of this single-template auction is the same as the truthful MITA auction, so since the SNE implements the former it implements the latter.
\end{proof}

While this provides a setting where positive results are possible, we note that it is somewhat special as one class has only a single slot.  
While limited generalisations appear possible when each class has multiple slots, they show a greatly reduced number of ads (essentially only those that would qualify to be shown if they were the first slot).

Another special case where we can provide positive results is one where each class corresponds to a block of ads, e.g. there could be a block of text ads somewhere on the page, a block of image ads, and a block of product ads.  
The assumption we make about these blocks is similar in spirit to the assumption in the traditional GSP model that the probability of an ad being clicked can be factored into a slot effect and an advertiser effect.  We assume that the slot effect can be further factored into a term that depends on the ordering within the block and a term that depends on where the block is shown.  Equivalently, we can phrase this in terms of the ratio of slot effects in different templates, as in the following definition.

\begin{definition}
A \emph{class selection} set of templates is one where whenever a class appears in multiple templates, the ratio of slot effects which are not zero are constant. More formally, if a class appears in a template with slot effects $s_1, s_2, ..., s_a$, and it appears in another template with slot effects $s'_1, s'_2, ..., s'_a$ then $s_i/s'_i$ is constant.
\end{definition}

Consider any monotone allocation rule that ranks bidders using a rank score and chooses the template based solely off the second highest ranked bidder of each class (a {\em second-highest} allocation rule).  Such rules could be relatively efficient in some settings, but in general cannot provide any bounded approximation to $OBJ$ because the value of the first bidder is ignored and can be arbitrarily high.  Nevertheless, we show that this allocation rule has an SNE with GSP payments.
\begin{theorem}
With a class selection template set with a second-highest allocation rule and template-indifferent GSP payment rule, an SNE is guaranteed to exist.
\end{theorem}
\begin{proof}
For every class take a template in which that class appears and then determine an SNE for that class in that template. Furthermore, alter the first bidder's bid to be equivalent (or slightly higher) than the second bidder's bid. Note that (for this class) this is not only still an SNE in the template, but by the nature of class selection templates the bids are then in SNE for all templates. Now upon determining the bids for all classes by doing this for every class, take the optimal template $T$ via the second-highest allocation rule. We claim that this is an SNE. This is because no single player can unilaterally deviate to raise the second highest bid of their class and the bidders are already in an SNE for the winning template $T$.
\end{proof}

Unfortunately, this positive result appears to rely both on restricting the setting (to class selection templates) as well as the allocation rule (to an inherently inefficient one).  The following theorem shows that with the standard allocation rule (optimizing $OBJ$) there may not exist an SNE for conservative bidders.  As the example is somewhat lengthy, we defer the proof to Appendix~\ref{sec:tigsp_nonexistence_proof}.
\begin{theorem}
\label{thm:tigsp_nonexistence}
Suppose we use the template-indifferent GSP payment rule with the standard allocation rule and furthermore, do not allow bidders to bid higher than their true value as before. Then there may not exist any SNE.
\end{theorem}
The example chosen in the proof is not carefully constructed. The sole essential ingredient is that the winning class $A$ bidders' true values are extremely close in value. This ensures that for every bidder, the range of bids allowed in an SNE is small. In fact, one can show via a similar proof that any class selection set such that every class appears in every template is enough to ensure there exists no SNE --- if the types/valuations of the bidders are chosen carefully.

\section{Discussion}

Our analysis examined trade-offs among different stakeholders in ad auctions. We have shown how to make these trade-offs optimally and how to handle constraints on them. We have also used simulations and real-world auction data to show the effectiveness of this approach, examining the impact of the allowed number of ads on the welfare of the interested parties.  In doing so, we demonstrated the efficacy of a per-impression reserve for this purpose, a tool which had previously been shown to have poor revenue-increasing properties.

We then examined a richer domain, where ads may be shown in several mutually exclusive templates.  Here our results are mostly negative, and we demonstrated the shortcomings of various natural generalizations of GSP.  While we were able to show some positive results, they were for restricted settings and allocation rules that did not achieve optimal trade-offs.

Our results on templates point to several directions for future research. First, are our examples for template auctions under natural GSP generalizations, which show arbitrarily bad or no equilibrium,  representative of real world ad-auctions?  In particular, a few of our constructions rely on a large number of bidders and slots.  Second, can good equilibria be found in ad auctions with templates under less demanding assumptions than we made to achieve our positive results? Finally, could a different generalization or set of assumptions lead to the existence of (approximately) optimal equilibria?

\begin{footnotesize}
\bibliography{REFERENCES}
\end{footnotesize}

\appendix

\section{Proofs of Section \ref{sec-constrained}}
\label{sec:constraint_proof}

Recall that we are assuming:
\begin{itemize}
\item The type space of allowable types $T$ is a closed compact set.
\item The allocation function $x: T \rightarrow \mathbb{R}^n$ is a Lebesgue integrable function $\in L^1$.
\item The constraint feasibility region $\mathcal{X}$ for the allocation function $x$ is convex w.r.t. the underlying function space, namely the  $L^1$  Banach space w.r.t. Lebesgue measure
\end{itemize}
In addition we shall assume that the linear system (w.r.t. the revenue, welfare, and click yield) of constraints is non-degenerate.

En route to proving Theorem \ref{thm:strongL}, we first show that two critical properties hold: weak duality and Lagrangian sufficiency.  We prove the results for equality constraints $A(x)=\theta$  (the same results hold true  for inequality constraints, $A(x) \geq \theta$ with the provision that $\lambda \geq 0$, and can be similarly proved).
The  Lagrangian in this case is 
\begin{equation*}
L(x, \lambda) = G(x) - \lambda^ T (\theta - A(x)).
\end{equation*}
The primal problem is
\begin{equation*}
\Gamma(\theta) = \sup_{x \in \mathcal{X}(\theta)} G(x)
\end{equation*}
where $\mathcal{X}(\theta) = \{x: x \in \mathcal{X} \text{ and } A(x) = \theta\}$ is the feasible set, and the dual is denoted
\begin{equation*}
h(\lambda) = \sup_{x \in \mathcal{X}} L(x, \lambda).
\end{equation*}

\begin{lemma}[Weak Duality]
\label{lem:weak}
For all $\lambda$, $h(\lambda) \geq \Gamma(\theta)$ and in particular for the optimal value of the dual, $h^* = \inf_{\lambda} h(\lambda) \geq \Gamma(\theta)$.
\end{lemma}
\begin{proof}

      $   h(\lambda) = \sup_{x \in \mathcal{X}} L(x, \lambda) \geq \sup_{x \in \mathcal{X}(\theta)}L(x, \lambda) =  \sup_{x \in \mathcal{X}(\theta)} G(x)=\Gamma(\theta)$
\end{proof}

\begin{lemma}[Lagrangian Sufficiency]
\label{lem:LagSuff}
If there is an $x^*$ and $\lambda^*$ such that
\begin{equation*}
L(x^*, \lambda^*) = \sup_{x \in \mathcal{X}} L(x, \lambda^*)
\end{equation*}
and $x^*$ is feasible for the constrained optimization (i.e. $x^* \in \mathcal{X}(\theta)$) then $x^*$ is optimal.
\end{lemma}
\begin{proof}
For  all  $x \in \mathcal{X}(\theta) $, we have $G(x)=L(x,\lambda)+\lambda^ T (\theta - A(x))$.  Since  $x^* \in \mathcal{X}(\theta)$,
\begin{equation*}
         G(x^*) =L(x^*, \lambda^*) =\sup_{x \in \mathcal{X}} L(x, \lambda^*) \geq L(x,  \lambda^*) =G(x)  \quad \mbox{for all } x \in \mathcal{X}(\theta) 
\end{equation*}
and hence  $x^*$  is optimal  and $G(x^*) =\Gamma(\theta)$.
\end{proof}

We now use the lemmas to prove Theorem \ref{thm:strongL}, restated here.

\begin{theorem}
If $\theta \in \relint(\dom(\Gamma))$, then the optimization problem is strong Lagrangian.
\end{theorem}
\begin{proof}
\begin{enumerate}
\item $G$ is concave (indeed, it is a linear functional), the constraint set $\mathcal{X}$ is convex, and the constraints involve a vector of linear functionals, hence we can show directly that $\Gamma$ is concave.\\
Assuming $\dom(\Gamma) \neq \varnothing$, take $\theta^1, \theta^2 \in \relint(\dom(\Gamma))$. Then there exists $x^1, x^2$ such that $G(x^i) = \Gamma(\theta^i)$ for $i \in \{1, 2\}$; these are feasible: $x^i \in \mathcal{X}(\theta^i)$. Now for any arbitrary $\lambda \in (0, 1)$ let $\tilde{\theta} = \lambda\theta^1 + (1 - \lambda)\theta^2$ and similarly $\tilde{x} = \lambda x^1 + (1 - \lambda)x^2$. By convexity of $\mathcal{X}$ and the linearity of constraints $A(x) = \theta$, we then have $\tilde{x} \in \mathcal{X}(\tilde{\theta})$. Thus:
\begin{align*}
\Gamma(\tilde{\theta})
&= \sup_{x \in \mathcal{X}(\tilde\theta)} G(x)\\
&\geq G(\tilde{x})\\
&= G(\lambda x^1 + (1 - \lambda)x^2)\\
&= \lambda_1 G(x^1) + (1 - \lambda)G(x^2)\\
&= \lambda \Gamma(\theta^1) + (1 - \lambda)\Gamma(\theta^2)
\end{align*}
and so $\Gamma$ is concave.

\item Since $\Gamma$ is concave, for $\theta \in \relint(\dom(\Gamma))$, there is a non-vertical  supporting hyperplane (NVSH) to $\Gamma$ at $\theta$  (which means that the hypograph of $\Gamma$ has a supporting hyperplane), which is equivalent to the problem being strong Lagrangian.
Note that $\Gamma: \theta \rightarrow \bR$ is defined on Euclidean space.   That such a problem is strong Lagrangian for $\Gamma(b)$ if and only if there is a NVSH at $b$ is a well known theorem, and simple to prove directly. For the ``if'' direction: suppose there exists a NVSH at $b$, then there is a $\lambda \in \bR^{\rho(A)}$ such that
\begin{equation*}
  \Gamma(c) \leq \Gamma(b) + \lambda \cdot (c-b)   \quad \mbox{for all }  c \in \bR^{\rho(A)}.
\end{equation*}
Hence
\begin{align*}
   \Gamma(b) & \geq   \sup_{ c \in \bR^{\rho(A)}} \left [\Gamma(c) -  \lambda \cdot (c-b)  \right ]\\
       & = \sup_{ c \in \bR^{\rho(A)}} \sup_{x \in \mathcal{X}(c)} \left [ G(x) - \lambda \cdot (A(x) - b)\right]\\
      & = \sup_{x \in \mathcal{X}} L(x,\lambda) \\
      & = h(\lambda).
\end{align*}
But from Weak Duality (Lemma \ref{lem:weak}) $\Gamma(b) \leq h(\lambda)$, and thus $\Gamma(b)=h(\lambda)$ and the problem is strong Lagrangian.  The ``only if'' is similarly proved by reversing the argument.
\end{enumerate}
\end{proof}


When the problem is strong Lagrangian, then $h^*=\Gamma(\theta)$, and complementary slackness applies.

We can now prove the equivalent of Slater's constraint qualification to derive  sufficient conditions to have $\theta \in \relint(\dom(\Gamma))$.   This is Corollary \ref{cor:Slater} which we restate and prove.  

\begin{corollary}
\label{cor:Slater}
  When $Z=\{z\in \bR^{\rho(A)}: \leq 0\} $ is the non-positive cone, the problem $X_{Z}(\theta)$, is strong Lagrangian if there is some $\tilde{x} \in \mathcal{X}$ such that $A(\tilde{x})>\theta$.  If $Z={0}$, the problem  $X_{Z}(\theta)=X(\theta)$ is strong Lagrangian provided there is some $\tilde{x} \in \relint(\mathcal{X})$ such that $A(\tilde{x})=\theta$.
\end{corollary}
\begin{proof}
When the conditions are satisfied, then clearly $\theta \in \relint(\dom(\Gamma))$, and hence the result.  Note that we do require  the condition when $Z=0$ because we have folded the constraints into $\mathcal{X}$, which means that is possible that $\theta \in \bdd \phi$ because $x \in \bdd \mathcal{X}$; when on the boundary, the separating hyperplane for $\Gamma$ at $\theta$ may be vertical, in which case  problem is not strong Lagrangian.  The condition when $Z=0$ precludes this from happening.
\end{proof}
%
%
%

\section{Proof of Theorem~\ref{thm:tcgsp_unoptimal}}
\label{sec:tcgsp_unoptimal_proof}

\begin{proof}
Suppose we have the following single-class setup where all the $w_i = 1$ and we are maximizing welfare.
\begin{itemize}
\item Bidders:
\begin{itemize}
\item $2m$ bidders with true values $m/2 - \epsilon, m/2 - 2\epsilon, ..., m/2 - (2m)\epsilon$.
\item $m^2$ bidders with true values $1 - \epsilon, 1 - 2\epsilon, ..., 1 - m^2\epsilon$.
\end{itemize}
\item Templates:
\begin{itemize}
\item Template 1: $2m + m^2 - 1$ slots of effect $1 - \epsilon, 1 - 2\epsilon, ..., 1 - 2m\epsilon, \epsilon, \epsilon^2, ..., \epsilon^{m^2 - 1}$.
\item Template 2: $2m + m^2 - 1$ slots of effect $x - \epsilon, x - 2\epsilon, ..., x - (2m + m^2 - 1)\epsilon$ where $x = 2(2m - 1)/(2m + m^2)$ which is equivalent to $x \approx 4/m$ for large $m$.
\end{itemize}
\end{itemize}
We first show that the first template will always win in any NE of the truthful auction for sufficiently large $m$ and small $\epsilon$\footnote{Specifically, we mean taking a large $m$ and \emph{then} taking a small $\epsilon$. For our purposes, it will suffice if $\epsilon \ll 1/m^2$.}. Recall that in an NE of the truthful auction, a bidder will always bid higher than the true value of the next highest bidder (as otherwise the next highest bidder would outbid them). As we are also assuming bidders cannot bid higher than their true value this implies that the first bidder (i.e. the bidder with valuation $m/2 - \epsilon$) will have bid $b_1 \in [m/2 - \epsilon, m/2 - 2\epsilon]$. More generally, bidder $i < 2m + m^2$ will bid $b_i \in [t_{i - 1}, t_i]$. We therefore find that the calculated objective of the first template subtracted by the second is (where calculations are in the limit for large $m$ and small $\epsilon$):
\begin{align*}
&(b_1 + b_2 + ... + b_{2m}) - (b_1 + b_2 + ... + b_{2m + m^2 - 1})x\\
&= (b_1 + b_2 + ... + b_{2m - 1})(1 - x) + b_{2m}(1 - x) - (b_{2m + 1} + b_{2m + 2} + ... + b_{2m + m^2 - 1})x\\
&= (2m - 1)(m/2)(1 - x) + b_{2m}(1 - x) - (m^2 - 1)x\\
&> (2m - 1)(m/2)(1 - x) - (m^2 - 1)x\\
&= m^2\\
&> 0.
\end{align*}
As this is greater than zero, the first template will always win for sufficiently large $m$ and small $\epsilon$.

We now show that there exists a problematic SNE when we use the template-considerate GSP payment rule where the second template wins. Suppose we have an SNE on the second template, ignoring the first. Assuming $\epsilon \ll x$, the slots in the second template are nearly identical. Thus, if we ignore the first template, the bidders will only bid enough to ensure they win a spot --- that is, enough to outbid the last bidder who has a true value of $1 - m^2\epsilon \approx 1$; with the caveat that the first bidder may bid up to their true valuation. We now claim this is an SNE overall, i.e. when we consider this in the context of {\em both} templates. A bidder who wishes to change to the first template will raise their bid to at most their true value. Let us consider the objective values of the two templates after this maximum unilateral deviation. The first template will have an objective value of at most $\approx 2(m/2) + (m/2 - 2)$ if we assume that one bidder is already bidding their true value. For large $m$, this is approximately $3m/2$. On the other hand, the second template will have objective value of at least $\approx (2m + m^2 - 1)x$ if all but the last bidder bid approximately $1$. For large $m$ this is approximately $4m$. As $4m > 3m/2$, we find that no unilateral deviation can switch the template and so we have an SNE overall.

Now note that this SNE has a welfare of $\approx ((2m)(m/2) + (m^2))x \in \theta(m)$. On the other hand, all NE of the truthful auction must use the first template, so their welfare is $\approx (2m)(m/2) \in \theta(m^2)$. Thus, this SNE is arbitrarily worse than any NE of the truthful auction.
\end{proof}

\section{Proof of Theorem~\ref{thm:tigsp_nonexistence}}
\label{sec:tigsp_nonexistence_proof}

\begin{proof}
Suppose we are maximizing welfare, all bidders have ad-effect $w_i = 1$ and we have the following setup.
\begin{itemize}
\item Classes $A$ and $B$.
\item Bidders:
\begin{itemize}
\item Four class $A$ bidders with true values $100$, $100 - \epsilon$, $100 - 2\epsilon$, and $20$.
\item Two class $B$ bidders with true values $150$ and $135$.
\end{itemize}
\item Templates ($\Delta \in (0, 1)$):
\begin{itemize}
\item Template $1$: three class $A$ slots of effects $1$, $1/2$, and $1/4$, and a class $B$ slot effect of $\Delta$.
\item Template $2$: three class $A$ slots of effects $\Delta$, $\Delta/2$, and $\Delta/4$, and a class $B$ slot effect of $1$.
\end{itemize}
Note that if $\Delta \approx 0$ we have a situation similar to our guiding example that is the image/text auction.
\end{itemize}
Let us consider the general properties of any SNE. Regarding class $B$ bidders, it is clear that the second bidder will always bid his true value of $135$, and the first bidder will outbid him. Now let us consider the far more complex class $A$ bidders. Let $b_i$ denote the bid of class $A$'s bidder $i$. Clearly $b_4 = 20$ as he will not win a slot and must therefore bid his true value to ensure he is not envious. We further find that:
\begin{itemize}
\item As the third bidder must not envy the second bidder:
\begin{align*}
&(100 - 2\epsilon - b_4)/4 \geq (100 - 2\epsilon - b_3)/2\\
&\Rightarrow b_3 \geq 60 - \epsilon.
\end{align*}
\item As the second bidder must not envy the third bidder:
\begin{align*}
&(100 - \epsilon - b_3)/2 \geq (100 - \epsilon - b_4)/4\\
&\Rightarrow b_3 \leq 60 - \epsilon/2.
\end{align*}
\item As the second bidder must not envy the first bidder:
\begin{align*}
&(100 - \epsilon - b_3)/2 \geq 100 - \epsilon - b_2\\
&\Rightarrow b_2 \geq 50 + b_3/2 - \epsilon/2.
\end{align*}
\item As the first bidder must not envy the second bidder:
\begin{align*}
&100 - b_2 \geq (100 - b_3)/2\\
&\Rightarrow b_2 \leq 50 + b_3/2.
\end{align*}
\end{itemize}
Together, these imply that:
\begin{align*}
b_3 &\in [60 - \epsilon, 60 - \epsilon/2],\\
b_2 &\in [50 + b_3/2 - \epsilon/2, 50 + b_3/2]\\
&\subseteq [50 + (60 - \epsilon)/2 - \epsilon/2, 50 + (60 - \epsilon/2)/2]\\
&= [80 - \epsilon, 80 - \epsilon/4].
\end{align*}
We now claim that in any SNE the second template is selected via the standard selection rule. To see this, let us consider the objective value of the first template subtracted by the second (when we assume $b_1 = b_2$).
\begin{align*}
&(b_2 + b_2/2 + b_3/4 + \Delta 135) - (\Delta(b_2 + b_2/2 + b_3/4) + 135)\\
&= (1 - \Delta)(b_2 + b_2/2 + b_3/4 - 135)\\
&\leq (1 - \Delta)(3(80 - \epsilon/4)/2 + (60 - \epsilon/2)/4 - 135)\\
&= -(1 - \Delta)(\epsilon/2)\\
&< 0
\end{align*}
where the first inequality comes from choosing $b_3$ and $b_2$ maximally. We next claim that for $\epsilon$ sufficiently small the third bidder can raise his bid $b_3$ so that it is beneficial to him in any SNE by causing the first template to be selected --- which would imply that there exists no SNE. To see this, let us consider the objective value of the first template subtracted by the second in an SNE when the third bidder unilaterally deviates and raises his bid to $70$.
\begin{align*}
&(b_2 + b_2/2 + 70/4 + \Delta 135) - (\Delta(b_2 + b_2/2 + 70/4) + 135)\\
&= (1 - \Delta)(3b_2/2 + 70/4 - 135)\\
&\geq (1 - \Delta)(3b_2/2 + 70/4 - 135)\\
&\geq (1 - \Delta)(3(80 - \epsilon)/2 + 70/4 - 135)\\
&\geq (1 - \Delta)(5/2 - 3\epsilon/2)\\
&> 0
\end{align*}
where the last inequality holds for sufficiently small $\epsilon$. As $70 < b_2$ for sufficiently small $\epsilon$, the third bidder will remain in the third position and his payment will be unaffected. Therefore, the third bidder will have incentive to raise his bid and so we cannot have an SNE.
\end{proof}

\end{document}